\documentclass[10pt,journal,onecolumn,compsoc]{IEEEtran}
 
\ifCLASSOPTIONcompsoc
  \usepackage[nocompress]{cite}
\else
  \usepackage{cite}
\fi

\usepackage{hyperref}
\hypersetup{hidelinks}       
\usepackage{url}            
\usepackage{booktabs}       
\usepackage{amsfonts}       
\usepackage{microtype}      
\usepackage{lipsum}       
\usepackage{graphicx}
\usepackage{doi}
\usepackage{amssymb}
\usepackage{amsmath}
\usepackage{amsthm}
\usepackage[lined,linesnumbered,ruled,norelsize]{algorithm2e}
\usepackage{algorithmicx}
\usepackage{color}
\usepackage{multicol}
\usepackage{subfigure}
\usepackage{setspace}
\usepackage{geometry}
\newtheorem{lemma}{Lemma}

\newtheorem{theorem}{Theorem}

\newtheorem{assumption}{Assumption}

\begin{document}


\title{Harnessing Context for Budget-Limited Crowdsensing with Massive Uncertain Workers}


\author{Feng~Li,
        Jichao~Zhao,
        Dongxiao~Yu,
        Xiuzhen~Cheng,
        and~Weifeng~Lv
\IEEEcompsocitemizethanks{\IEEEcompsocthanksitem F. Li, J. Zhao, D. Yu and X. Cheng are with School of Computer Science and Technology, Shandong University, Qingdao 266237, China.\protect\\
E-mail: \{fli, dxyu, xzcheng\}@sdu.edu.cn, zhaojichao@mail.sdu.edu.cn
\IEEEcompsocthanksitem W. Lv is with School of Computer Science and Engineering, Beihang University, Beijing 100191, China. \protect\\
E-mail: lwf@nlsde.buaa.edu.cn
}
}

\IEEEtitleabstractindextext{%
\begin{abstract}
  Crowdsensing is an emerging paradigm of ubiquitous sensing, through which a crowd of workers are recruited to perform sensing tasks collaboratively. Although it has stimulated many applications, an open fundamental problem is how to select among a massive number of workers to perform a given sensing task under a limited budget. Nevertheless, due to the proliferation of smart devices equipped with various sensors, it is very difficult to profile the workers in terms of sensing ability. Although the uncertainties of the workers can be addressed by conventional \textit{Combinatorial Multi-Armed Bandit} (CMAB) framework through a trade-off between exploration and exploitation, we do not have sufficient allowance to directly explore and exploit the workers under the limited budget. Furthermore, since the sensor devices usually have quite limited resources, the workers may have bounded capabilities to perform the sensing task only few times, which further restricts our opportunities to learn the uncertainty. To address the above issues, we propose a \textit{Context-Aware Worker Selection} (CAWS) algorithm in this paper. By leveraging the correlation between the context information of the workers and their sensing abilities, CAWS aims at maximizing the expected cumulative sensing revenue efficiently with both budget constraint and capacity constraints respected, even when the number of the uncertain workers are massive. The efficacy of CAWS can be verified by rigorous theoretical analysis and extensive experiments.

\end{abstract}

\begin{IEEEkeywords}
Multi-Armed Bandits, worker selection, crowdsensing
\end{IEEEkeywords}}

\maketitle

\IEEEdisplaynontitleabstractindextext
\IEEEpeerreviewmaketitle

\section{Introduction} \label{sec:intro}
  Due to the proliferation of hand-held smart devices (e.g., smart phones, smart glasses, smart watches, etc) which are usually equipped with various sensors, the concept of crowdsensing has become a new paradigm for ubiquitous sensing \cite{KhanXAA-SURV13,CapponiFKFKB-COMST19}. Thousands or even millions of human crowds (a.k.a., \textit{workers}) can be engaged in a sensing task (e.g., traffic information collection, air quality surveillance, urban business survey, urban WiFi characterization, etc) with their sensor devices, and their collective contributions can be utilized to considerably improve the sensing quality across a wide spectrum of applications~\cite{CherianLGHW-MDM16,WangTLP-INFOCOM17,LiuLZMZ-IMWUT18,JiangZZLSMFWL-TMC21}.

  Although we have no need to deploy specialized sensor devices to conduct a sensing task by applying the crowdsensing paradigm and the overhead of data acquisition is thus considerably reduced, the requester of the sensing task is still constrained by a budget such that the requester only affords to recruit a limited number of workers. Therefore, how to select a subset of high-qualified workers from a large crowd is a very crucial issue for guaranteeing the accomplishment of the sensing task. There have been many existing studies exploring the combinatorial nature of the worker selection problem by assuming the workers' sensing abilities are known in advance \cite{SongLWMW-TVT14,LiLW-MASS15,PuCXF-INFOCOM16,YangLWW-TMC19}.

  Due to the diversities of sensor devices and human behaviors, workers may have distinct sensing abilities to provide data with different qualities even for the same sensing task, while it is usually very difficult to pre-profile the heterogeneous workers in terms of sensing ability, especially when the number of the workers may be huge. To address the uncertainties of the workers, one popular choice is to apply the \textit{Combinatorial Multi-Armed Bandits} (CMAB) framework such that the workers are sequentially selected to perform the sensing task under a budget and the performance of the workers in conducting the sensing task can be observed to estimate their sensing abilities. Existing proposals usually leverage a trade-off between exploration and exploitation for each of the workers~\cite{HanZL-TON16,RangiF-AAMAS18,ZhaoXWXHZ-TMC21,SongJ-INFOCOM21} (or each of the worker-task combinations~\cite{GaoWYXC-ICPADS19,GaoWXC-INFOCOM20}); therefore, such conventional CMAB-based approaches result in significant overhead and thus are of low scalability, especially when the budget is heavily limited whereas the number of the unknown workers is huge. For an extreme example, if we do not have a sufficient budget to select each of the workers once, these approaches based on the conventional CMAB framework even cannot be initialized if directly exploring and exploiting the individual workers. In addition, the workers may have bounded sensing capabilities due to the resource limit of their sensor devices; hence, each of the workers conducts the sensing tasks only for a few times, such that we may not have sufficient opportunities to explore and exploit them individually. In a nutshell, we focus on addressing the following open problem in this paper: \textit{given a budget-limited sensing task, how to fully utilize the budget to efficiently select among a massive number of unknown workers with bounded capacities through an exploration-exploitation trade-off?}

  In this paper, we propose a \textit{Context-Aware Worker Selection} (CAWS) algorithm. Specifically, inspired by the fact that workers with similar contexts usually have similar sensing abilities, we innovate in adapting the CMAB framework to learn the correlation between the workers' contexts and their sensing abilities rather than the sensing abilities of the individual workers. Through partitioning the context space into a set of sub-spaces (i.e., the so-called ``hypercubes'') with a fine-tuned granularity, we can learn the sensing ability distribution upon the hypercubes and thus efficiently estimate the sensing abilities of the workers with similar contexts in each hypercube. According to the estimates, we enable an efficient exploration-exploitation trade-off to select among massive unknown workers with bounded capacities under a limited budget in the context space. By our CAWS algorithm, the expected cumulative sensing revenue can be maximized with both the budget constraint and the capacity constraints respected. We conduct a solid theoretical analysis to quantify the performance gap (a.k.a. \textit{regret}) between our algorithm and the (nearly) optimal off-line algorithm where the workers' sensing abilities are known as prior. We also perform extensive experiments on both synthetic and real datasets to verify the efficacy of our CAWS algorithm. The main contribution of this paper is summarized as follows.
  \begin{itemize}
    \item To the best of our knowledge, this is the first work considering the scalability in efficiently selecting among a massive number of unknown workers under a significantly limited budget.
    \item We propose a context-aware worker selection algorithm to maximize the expected cumulative sensing revenue with both the budget constraint and the capacity constraints respected.
    \item We conduct a rigorous theoretical analysis to quantify the regret between our CAWS algorithm and the approximately optimal one, and perform extensive experiments on both synthetic data and real data to verify the advantages of CAWS over other state-of-the-art methods.
  \end{itemize}

  The remaining of our paper is organized as follows. We first introduce our system model and describe our problem in Sec.~\ref{sec:sys}. We then present the details of our CAWS algorithm in Sec.~\ref{sec:algo}. The analysis of our CAWS algorithm is given in Sec.~\ref{sec:analysis}. We report our experiment results in Sec.~\ref{sec:exp}. We finally survey related literature and conclude this paper in Sec.~\ref{sec:relwork} and Sec.~\ref{sec:conclusion}, respectively.

\section{System Model and Problem Description} \label{sec:sys}
  \subsection{System Model} \label{ssec:model}
    We consider a crowdsensing process assigning a sensing task to a set of workers $\mathcal N=\{1,2,\cdots, N\}$ under budget $B$. For each worker $i \in \mathcal N$, let $c_i$ denote the cost to recruit (or select) worker $i$ for one time to collect a data sample. For example, a worker should be paid when it is selected to report a data sample. Moreover, the workers may carry various sensor devices with different configurations (e.g., in communication modules, sampling resolutions, etc.); hence the cost parameters for the different workers are distinct. We assume $c_{min} = \min_{i\in\mathcal N}c_i$ and $c_{max} = \max_{i\in\mathcal N}c_i$. We also define a capacity attribute $\tau_i$ for each worker $i\in \mathcal N$, which represents the maximum number of data samples worker $i$ can contribute (or the maximum number of times worker $i$ can be selected) due to the resource limit of its sensor devices. Let $\tau_{max}= \max_{i\in \mathcal N} \tau_i$ be the maximum capacity.

    Each worker $i\in\mathcal N$ is associated with context information denoted by $\phi_i$ which is closely related to the worker's sensing ability. We assume that, for $\forall i \in \mathcal N$, $\phi_i \in \mathcal S$ is an $M$-dimensional vector, where $\mathcal S = [0,1]^M$ is the so-called ``\textit{context space}''. The context dimensions could include the proficiency of the workers in some required skills, the personal backgrounds of the workers or the performance parameters of the sensor devices. We can normalize each of the dimensions into a range of $[0,1]$. We define a \textit{stochastic} reward function $r: \mathcal S \rightarrow \{0,1\}$ upon the context space $\mathcal S$. For $\forall i \in \mathcal N$, the binary random variable $r(\phi_i) \in \{0,1\}$ indicates if a data sample provided by worker $i$ is qualified and thus represents the random \textit{reward} obtained by selecting worker $i$ to collect a data sample. We assume $r(\phi_i)$ for each selection (and thus for each data sample) is identically and independently drawn from an \textit{unknown} Bernoulli distribution and let $\mu_i = \mathbb E[r(\phi_i)]$ denote the unknown expectation of $r(\phi_i)$ \footnote{Although we hereby assume $r(\phi_i)$ is an i.i.d. random variable obeying an unknown Bernoulli distribution parameterized by $\mu_i = \mathbb P(r(\phi_i)= 1)$, our algorithm readily works with arbitrary probability distributions with normalized supports in $[0,1]$.}. In fact, $\mu_i$ is a measure of worker $i$'s sensing ability. To facilitate our presentation, we suppose $r_i = r(\phi_i)$ and thus $\mu_i = \mathbb E[r_i]$ throughout the remaining of this paper.

  \subsection{Problem Description} \label{ssec:form}
    Assuming $x_i\in \{ 0,1,\cdots,\tau_i \}$ is the number of times we select worker $i$ (i.e., the number of data we collect through recruiting worker $i$), our problem can be formulated as
    \begin{eqnarray} 
      &\max& ~~ f(\{x_i\}^N_{i=1}) = \sum^N_{i=1} \mu_i x_i \label{eq:obj}\\
      &\mathrm{s.t.}& ~~ \sum^N_{i=1} x_i c_i \leq B, \label{eq:budget}\\
      && ~~ x_i \in \left\{ 0,1,2,\cdots,\tau_i \right\}, ~\forall i\in\mathcal N \label{eq:cap}
    \end{eqnarray}
    Our objective (\ref{eq:obj}) is to maximize the expected cumulative revenue induced by our task assignment $\{x_i\}^N_{i=1}$, subject to budget constraint (\ref{eq:budget}) and capacity constraints (\ref{eq:cap}). In particular, the total cost of our task assignment cannot exceed the budget and each worker cannot be selected more than $\tau_i$ times. The problem formulation actually characterizes a general crowdsensing scenario, as illustrated in many existing studies, e.g., \cite{HanZL-TON16,RangiF-AAMAS18,ZhaoXWXHZ-TMC21,SongJ-INFOCOM21}.

    It is apparent that, if $\mu_i$ (or $r(\cdot)$) was known as prior knowledge, our problem could be cast to a \textit{Bounded Knapsack Problem} (BKP). Although the BKP is of NP-hardness, it can be addressed by many approximation algorithms efficiently \cite{KohliKM-EJOR14}. For example, in the $2$-approximation density-order greedy algorithm, we first sort the workers in decreasing order according to their densities $\rho_i = {\mu_i}/{c_i}$, and then greedily select the workers in the order until we do not have sufficient residual budget to select any available worker with non-zero residual capacity. As will be shown in Sec.~\ref{sec:algo}, we adapt this algorithm as a subroutine in our CAWS algorithm, where we sort the workers according to the estimates on their densities.

    Unfortunately, it is usually very difficult to pre-profile the workers due to the huge number of workers as well as the diversity of the sensor devices carried by the workers. Consequently, $\{\mu_i\}^N_{i=1}$ may not always be available as prior, which makes our problem much more difficult than the BKP. To address such uncertainties, one choice is to apply the CMAB framework. For example, in \cite{RangiF-AAMAS18}, the workers (corresponding to the arms) are explored and exploited through UCB indexing. Nevertheless, when there are a huge number of workers (and thus arms), leveraging the trade-off between exploration and exploitation directly among the workers results in considerable overhead. For example, in an extreme case where $\sum^N_{i=1}c_i > B$, we even do not have sufficient budget to select each of the workers for one time to initialize the workers' UCB indices. 
    %
    %
    Therefore, the problem is, \textit{given a massive number of workers with unknown sensing abilities, how to efficiently select among them to maximize the expected cumulative sensing revenue with both the budget constraint and the capacity constraints respected?} In this paper, we propose to utilize the correlation between context information and sensing ability, for the purpose of balancing exploration and exploitation among the workers in the context space.

\section{Algorithm} \label{sec:algo}
  Our CAWS algorithm is motivated by a common sense that workers with similar context may have similar sensing abilities for a certain type of sensing tasks (which is the main basis for our later theoretical analysis). We divide the context space $\mathcal S$ into $d^M$ disjoint cubic sub-spaces (which are called ``\textit{hypercubes}'' in the following). Each of the $M$-dimensional hypercubes is of identical size $\frac{1}{d} \times \frac{1}{d} \times \cdots \times \frac{1}{d}$ \footnote{We will introduce how to partition the context space by choosing a proper value for $d$ later in Sec.~\ref{sec:analysis}.}. We denote by $\Omega$ the set of all hypercubes and by $Q_i \in \Omega$ the one such that $\phi_i \in Q_i$. As mentioned above, the workers in the same hypercube may have similar sensing abilities. Therefore, the essence of our CAWS algorithm is to leverage the trade-off between exploration and exploitation among the hypercubes in the context space rather than the individual workers. By learning the ``sensing abilities'' of the hypercubes, we can estimate the ones of the workers according to their contexts.

  The pseudo-code of our CAWS algorithm is described in \textbf{Algorithm}~\ref{alg:contextalgo}. Our algorithm proceeds in iterations. We denote by $i(t) \in \mathcal N$ the worker selected in the $t$-th iteration and by $r_{i(t)}$ the random reward yielded by this selection. For $\forall Q \in \Omega$, it is said that we choose $Q$ in the $t$-th iteration if $\phi_{i(t)} \in Q$. We then denote by 
  \begin{equation} \label{eq:chtimes}
    \lambda_Q(t) = \sum^t_{t'=1} \mathbb I(\phi_{i(t')} \in Q) = \lambda_Q(t-1) + \mathbb I(\phi_{i(t)} \in Q)
  \end{equation}
  the number of times $Q$ is chosen up to the $t$-th iteration, where $\mathbb I: \{\mathrm{True}, \mathrm{False}\} \rightarrow \{1,0\}$ is an indicator function. We also denote by 
  \begin{align} \label{eq:avgreward}
    \bar r_{Q}(t) &= \frac{\sum^t_{t'=1} \mathbb I(\phi_{i(t')}\in Q) r_{i(t')}}{\lambda_Q(t)} \nonumber\\
    &= \frac{\bar r_{Q}(t-1)\lambda_Q(t-1) + \mathbb I(\phi_{i(t)}\in Q) r_{i(t)}}{\lambda_Q(t)}
  \end{align}
  the average reward obtained up to the $t$-th iteration by choosing $Q$. At the beginning of the $t$-th iteration, we let $B(t)$ be the residual budget and $\tau_i(t)$ be the residual capacity of worker $i \in \mathcal N$, which are initialized by $B(1)=B$ and by $\tau_i(1) = \tau_i$, respectively, as shown in Line~\ref{ln:init0}. Worker $i\in\mathcal N$ is said to be \textit{available} in the $t$-th iteration if $\tau_i(t) \geq 1$. Our algorithm proceeds only if there exists sufficient budget to select at least one available worker (see Line~\ref{ln:available}). In the first $d^M$ iterations, we randomly choose a worker from each of the hypercubes, so as to initialize $\lambda_Q(t)$ and $\bar r_Q(t)$ for $\forall Q$ (see Lines \ref{ln:first} and \ref{ln:init1}). In the following, we use a density-ordered greedy subroutine (see \textbf{Algorithm}~\ref{alg:sub}) to calculate a non-negative integral weight $x_i(t)$ for $\forall i \in \mathcal N$, which represents how many times we could (virtually) select worker $i$ using residual budget $B(t)$ in a greedy manner (see Line~\ref{ln:sub}). We then choose worker $i(t)$ with probability $\frac{x_i(t)}{\sum^N_{i'=1}x_{i'}(t)}$ (see Line~\ref{ln:worker}) and increase $x_{i(t)}$ by one accordingly (see Line~\ref{ln:setx}). Next, we update $\bar r_{Q_{i(t)}}(t)$ and $\lambda_{Q_{i(t)}}(t)$ for the hypercube $Q_{i(t)}$ (see Line~\ref{ln:updatereward}). We finally renew the residual capacity of $i(t)$ and the residual budget (as shown in Lines~\ref{ln:updatecap} and \ref{ln:updatebudget}, respectively) and proceed to the next iteration (see Line~\ref{ln:updateit}). 
    \begin{algorithm}[htb!]
        \KwIn{$\{\tau_i, c_i, \phi_i\}^N_{i=1}$, $B$}
        \KwOut{$\mathbf{x} = \{x_i\}^N_{i=N}$}
        $t=1$; $B(t) = B$; $\tau_i(t) = \tau_i$ and $x_i = 0$ for $\forall i \in \mathcal N$; \label{ln:init0}\\
        \While{$B(t) \geq \min\{c_i \mid i\in \mathcal N, \tau_i(t) \geq 1\}$}{ \label{ln:available}
        \eIf{$t \leq d^M$ }{ \label{ln:first}
            Randomly choose worker $i(t)$ in the $t$-th hypercube; \label{ln:init1}\\
        }
        {
            Call the density-ordered greedy subroutine (see \textbf{Algorithm}~\ref{alg:sub}) to calculate $\{x_i(t)\}^N_{i=1}$; \label{ln:sub}\\
            Choose worker $i(t) \in \mathcal N$ with probability $\frac{x_{i}(t)}{\sum_{i'\in\mathcal N} x_{i'}(t)}$; \label{ln:worker}\\
            $x_{i(t)} = x_{i(t)} + 1$; \label{ln:setx}\\
        }
        Observe $r_{i(t)}$; \label{ln:reward}\\
        Update $\lambda_{Q_{i(t)}}(t)$ and $\bar r_{Q_{i(t)}}(t)$ according to (\ref{eq:chtimes}) and (\ref{eq:avgreward}), respectively; \label{ln:updatereward}\\
        $\tau_{i(t)}(t+1) = \tau_{i(t)}(t) - 1$; \label{ln:updatecap}\\
        $B(t+1) = B(t) - c_{i(t)}$; \label{ln:updatebudget}\\
        $t = t+1$; \label{ln:updateit}\\
      }
      \caption{Our context-aware worker selection algorithm.} 
    \label{alg:contextalgo}
    \end{algorithm}

  As demonstrated in \textbf{Algorithm}~\ref{alg:contextalgo}, a density-ordered greedy subroutine is called in each iteration to calculate $x_i(t)$. The pseudo-code of the subroutine is given in \textbf{Algorithm}~\ref{alg:sub}. Specifically, in the $t$-th iteration, we first calculate UCB index
  \begin{equation} \label{eq:ucb}
    U_i(t) = \bar r_{Q_i}(t-1) + \sqrt{\frac{2\log t}{\lambda_{Q_i}(t-1)}}
  \end{equation}
  for each worker $i$ (see Line~\ref{ln:ucb}), and the workers are then sorted in decreasing order with respect to $\rho_i(t) = {U_i(t)}/{c_i}$. The UCB index $U_i(t)$ actually can be thought as an estimate on worker $i$'s sensing ability. We greedily choose among the workers with budget $B(t)$ in the order until there is no available worker or the residual budget is not sufficient for us to select any available workers (see Lines \ref{ln:initb}$\sim$\ref{ln:end}). 
    \begin{algorithm}[htb!]
    \KwIn{$\{\tau_i(t), c_i\}^N_{i=1}$, $B(t)$, $\{\bar r_Q(t-1), \lambda_Q(t-1)\}_{Q\in \Omega}$}
    \KwOut{$\mathbf{x}(t)=\{x_i(t)\}_{i\in\mathcal N}$}
      Calculate $U_i(t)$ for $\forall i \in \mathcal N$ according to (\ref{eq:ucb}); \label{ln:ucb}\\
      Sort the workers $\mathcal N$ in decreasing order with respect to $\rho_i(t) = {U_i(t)}/{c_i}$; \label{ln:order}\\
      $b=0$; \label{ln:initb}\\
      \For{$i=1, 2, \cdots, N$}{
        \eIf{$b+c_{i} \leq B(t)$}{
          $x_{i}(t) = \min \left\{\tau_{i}(t), \left\lfloor \frac{B(t)-b}{c_{i}} \right\rfloor\right\}$; \label{ln:calx1}\\
          $b = b + c_{i} \cdot x_{i}(t)$; \label{ln:updateb}\\
        }
        {
          $x_{i}(t) = 0$; \label{ln:calx2}\\
        }
      }\label{ln:end}
    \caption{Density-ordered greedy subroutine in the $t$-th iteration.}
    \label{alg:sub}
    \end{algorithm}

\section{Analysis} \label{sec:analysis}
  As mentioned in Sec.~\ref{ssec:form}, the BKP (\ref{eq:obj})$\sim$(\ref{eq:cap}) is NP-hard even when $\{\mu_i\}^N_{i=1}$ are known as prior. We now introduce a \textit{rounding}-based approximation algorithm which can serve as a baseline to theoretically evaluate our CAWS algorithm. In particular, we first fractionalize the (integral) BKP as follows
  \begin{align}
    &\max ~~ f(\{x_i\}^N_{i=1}) = \sum^N_{i=1} \mu_i x_i  \label{eq:fkobj}\\
    &\mathrm{s.t.} ~~ \sum^N_{i=1} x_i c_i \leq B, ~0 \leq x_i \leq \tau_i, ~\forall i\in\mathcal N \label{eq:fkcon}
  \end{align}
  and then round the fractional solution to an integral one. Compared with the (integral) BKP (\ref{eq:obj})$\sim$(\ref{eq:cap}), the only difference between them is that the variable $x_i$ is a fractional non-negative number in the \textit{Fractional BKP} (FBKP) rather than an integral non-negative number in the BKP. The FBKP can be addressed by a density-ordered greedy approach. Specifically, we first sort the workers in decreasing order with respect to their densities $\rho_i = \mu_i / c_i$ such that $\rho_1 \geq \rho_2 \geq \cdots \geq \rho_N$. Then, the optimal solution to the FBKP can be calculated as
    \begin{equation} \label{eq:fbkpsolu}
      x^*_i = \begin{cases}
        \tau_i, ~~\forall i=1,2,\cdots,k-1 \\
        \frac{B-\sum^{k-1}_{j=1}c_{j} \tau_j}{c_i}, ~~i=k \\
        0, ~~\forall i=k+1, k+2, \cdots, N
      \end{cases}
    \end{equation}
    where the $k$-th worker is continuously ``split'' such that $\sum^{k-1}_{j=1} c_{j}\tau_{j} \leq B$ and $\sum^{k}_{j=1} c_{j}\tau_{j} > B$. We finally round downward $x^*_i$ for $\forall i\in \mathcal N$, and denote by $\lfloor\mathbf x^*\rfloor = \{\lfloor x^*_i \rfloor\}^N_{i=1}$ the resulting integral solution to the BKP. Letting $f^*_{BKP}$ and $f^*_{FBKP}$ be the optimal objective value of the BKP and the one of the FBKP, respectively, we have
    \begin{equation} \label{eq:optsoluinq}
      \sum^N_{i=1} \mu_i \lfloor x^*_i \rfloor \leq f^*_{BKP} \leq f^*_{FBKP} \leq \sum^N_{i=1} \mu_i \lfloor x^*_i \rfloor+\mu_k
    \end{equation}
    It is shown that the gap between the lower-bound and the upper-bound of $f^*_{BKP}$ is bounded; hence, it is rational to use the lower-bound $\sum^N_{i=1} \mu_i \lfloor x^*_i \rfloor$ as the baseline to evaluate the performance of our algorithm. Supposing the budget is exhausted in $T$ iterations by our CAWS algorithm and $\{i(t)\}^T_{t=1}$ are the selected workers within the $T$ iterations, we are interested in investigating the upper-bound of the following regret function
    \begin{align} \label{eq:reg}
      \mathsf{Regret}(T, \{i(t)\}^T_{t=1})= \sum^N_{j=1} \mu_j \lfloor x^*_j \rfloor -  \sum^N_{j=1} \mu_j \mathbb E_{T, \{i(t)\}^T_{t=1}} \left[ x_i \right]
    \end{align}
    which indicates the gap between the expected cumulative revenue yielded by the (nearly) optimal solution $\lfloor \mathbf{x}^* \rfloor$ and the one produced by our solution $\mathbf{x}$.
    
    As mentioned in Sec.~\ref{sec:algo}, our CAWS algorithm is based on the natural assumption that workers with similar context could have similar sensing abilities. This assumption can be formalized by the following H$\ddot{\text{o}}$lder condition.
    \begin{assumption}[H$\ddot{\text{o}}$lder Condition] \label{ap:holder}
      When there exist $L > 0$ and $\alpha>0$ such that for any contexts $s, s' \in \mathcal S$, it holds that
      \begin{equation}
       \left|\mathbb E[r(s)] - \mathbb E[[r(s')] \right| \leq L \|s - s'\|^\alpha
      \end{equation}
      where $\|\cdot\|$ denotes the Euclidean norm in $\mathbb R^M$.
    \end{assumption}
    \noindent It should be noted that our CAWS algorithm still works if the assumption does not strictly hold. However, the regret might not be bounded if the assumption was violated.
    \begin{lemma} \label{le:errincube}
      For $\forall i, i' \in \mathcal N$ such that $Q_i = Q_{i'}$, we have 
      \begin{eqnarray} \label{eq:errincube}
        |\mu_{i} - \mu_{i'}| \leq \Delta =L \left( M^{\frac{1}{2}}d^{-1} \right)^\alpha
      \end{eqnarray}
    \end{lemma}
    \begin{proof}
    Since the workers $i$ and $i'$ have their contexts in the same hypercube, we have $\|\phi_{i} - \phi_{i'}\| \leq M^{\frac{1}{2}}d^{-1}$ according to our strategy of evenly partitioning the context space (as depicted in Sec.~\ref{sec:algo}). Then, considering the H$\ddot{\text{o}}$lder condition shown in \textbf{Assumption}~\ref{ap:holder}, we have $|\mu_i - \mu_{i'}| = |\mathbb E[r(\phi_i)] - \mathbb E[r(\phi_{i'})]| \leq L\|\phi_{i}-\phi_{i'}\|^\alpha = L \left( M^{\frac{1}{2}}d^{-1} \right)^\alpha$
    \end{proof}
    \noindent For each hypercube $Q \in \Omega$, we denote by $\mu_Q$ the expected reward yielded by randomly selecting any worker with its context in $Q$ (i.e., the ``sensing ability'' of the hypercube $Q$). It is apparent that $|\mu_i - \mu_{Q}| \leq \Delta$ for $\forall i \in \mathcal N$ such that $\phi_i \in Q$, which implies $\mu_{Q_i}$ can be used as an estimate on $\mu_i$. As will be shown later, \textbf{Lemma}~\ref{le:errincube} is one of the bases for the decomposition of our regret function (\ref{eq:reg}).

    The definition of the regret function suggests the key of our analysis should be to quantify the impact of mischoosing the workers on the sensing revenue. The main reasons for the regret are two-fold: on one hand, we use the qualities of the contextual hypercubes to estimate the ones of the workers such that we may not be able to make ``right'' selection decisions even we learn $\mu_Q$ exactly according to \textbf{Lemma}~\ref{le:errincube}; on the other hand, according to MAB theory, we learn the qualities of the contextual hypercubes through an exploration-exploitation trade-off, while making ``wrong'' selection decisions is the price we have to pay in the learning process. Therefore, supposing $\lfloor \tilde{\mathbf{x}}^* \rfloor = \{\lfloor \tilde{x}^*_i \rfloor\}^N_{i=1}$ is the solution obtained by applying the rounding-based density-ordered greedy algorithm to BKP instance $(\{i, \mu_{Q_i}, c_i, \tau_i\}^N_{i=1}, B)$ (where we use $\mu_{Q_i}$ as an estimate on $\mu_i$), we decompose the regret function as follows
    \begin{align} 
      & \mathsf{Regret}(T, \{i(t)\}^T_{t=1}) \nonumber\\
      =& \sum^N_{j=1} \mu_j \lfloor x^*_j \rfloor - \sum^N_{j=1}\mu_j \lfloor \tilde{x}^*_j \rfloor  + \sum^N_{j=1}\mu_j \lfloor \tilde{x}^*_j \rfloor - \sum^N_{j=1} \mu_j \mathbb E_{T, \{i(t)\}^T_{t=1}} \left[ x_j \right] \nonumber\\
      \leq& \sum^N_{j=1} \mu_j \lfloor x^*_j \rfloor - \sum^N_{j=1}\mu_j \lfloor \tilde{x}^*_j \rfloor  + \sum^N_{j=1} (\mu_{Q_j}+\Delta) \lfloor \tilde{x}^*_j \rfloor - \sum^N_{j=1} (\mu_{Q_j}-\Delta) \mathbb E_{T, \{i(t)\}^T_{t=1}} \left[ x_j \right] \nonumber\\
      \leq& \sum^N_{j=1} \mu_j \lfloor x^*_j \rfloor - \sum^N_{j=1}\mu_j \lfloor \tilde{x}^*_j \rfloor + \sum^N_{j=1} \mu_{Q_j} \lfloor \tilde{x}^*_j \rfloor - \sum^N_{j=1} \mu_{Q_j} \mathbb E_{T, \{i(t)\}^T_{t=1}} \left[ x_j \right] + \frac{2B\Delta}{c_{min}}
    \end{align}
    where we have the second inequality since $|\mu_j - \mu_{Q_j}| \leq \Delta$ holds for $\forall j \in \mathcal N$ as shown in \textbf{Lemma}~\ref{le:errincube} and the third one due to the fact that $\sum^N_{j=1} \lfloor \tilde{x}^*_j \rfloor \leq \frac{B}{c_{min}}$ and $\sum^N_{j=1} \mathbb E_{T, \{i(t)\}^T_{t=1}} \left[ x_j \right] \leq \frac{B}{c_{min}}$. By defining
    %
    %
    \begin{equation} \label{eq:reg1}
      \mathsf{Regret}(\lfloor \mathbf{x}^* \rfloor, \lfloor \tilde{\mathbf x}^* \rfloor, \{\mu_i\}^N_{i=1}) = \sum^N_{j=1} \mu_j \lfloor x^*_j \rfloor - \sum^N_{j=1}\mu_j \lfloor \tilde{x}^*_j \rfloor
    \end{equation}
    and
    %
    \begin{align} \label{eq:reg2}
      & \mathsf{Regret}(\lfloor \tilde{\mathbf x}^* \rfloor, \mathbf x, \{\mu_{Q_i}\}^N_{i=1}) = \sum^N_{j=1} \mu_{Q_j} \lfloor \tilde{x}^*_i \rfloor - \sum^N_{j=1} \mu_{Q_j} \mathbb E_{T, \{i(t)\}^T_{t=1}} \left[ x_j \right]
    \end{align}
    the regret function $\mathsf{Regret}(T, \{i(t)\}^T_{t=1})$ can be re-written as
    %
    \begin{align} \label{eq:regdecom}
      & \mathsf{Regret}(T, \{i(t)\}^T_{t=1}) \leq \mathsf{Regret}(\lfloor \mathbf{x^*} \rfloor, \lfloor \tilde{\mathbf x}^* \rfloor, \{\mu_i\}^N_{i=1}) + \mathsf{Regret}(\lfloor \tilde{\mathbf x}^* \rfloor, \mathbf x, \{\mu_{Q_i}\}^N_{i=1}) + \frac{2B\Delta}{c_{min}}
    \end{align}
    $\mathsf{Regret}(\lfloor \mathbf{x^*} \rfloor, \lfloor \tilde{\mathbf x}^* \rfloor, \{\mu_i\}^N_{i=1})$ represents the loss due to our estimation on the workers' sensing abilities through partitioning the context space, while $\mathsf{Regret}(\lfloor \tilde{\mathbf x}^* \rfloor, \mathbf x, \{\mu_{Q_i}\}^N_{i=1})$ indicates the one resulting from our exploration-exploitation trade-off to learn the qualities of the hypercubes. In the following, we first present the main result showing the upper-bound of the regret function (\ref{eq:regdecom}) in \textbf{Theorem}~\ref{thm:main} (see Sec.~\ref{ssec:main}) and then report the details of the proof in Sec.~\ref{ssec:proof}, where the two sub-regret functions (\ref{eq:reg1}) and (\ref{eq:reg2}) are bounded in \textbf{Theorem}~\ref{thm:reg1bd} and \textbf{Theorem}~\ref{thm:reg2upbd}, respectively.

    To facilitate our analysis, we reuse the notion $\in$ when doing so will not induce any  ambiguity. In particular, for each worker $i$, it is said that $i\in \lfloor \mathbf{x}^* \rfloor$ (resp. $i\in \lfloor \tilde{\mathbf{x}}^* \rfloor$) if $\lfloor {x}^*_i \rfloor \geq 1$ (resp. $\lfloor \tilde{x}^*_i \rfloor \geq 1$). We also give some notions as follows which will be useful to our later analysis.
    \begin{eqnarray}
      && i^* = \arg\max_{i\in\mathcal N}\frac{\mu_{Q_i}}{c_i} \\
      && \mathcal{N}_Q = \{i \in \mathcal N \mid \phi_i \in Q\} \\
      && \mathcal{N}^+_Q = \{i \in \mathcal N \mid \phi_i \in Q, i \in \lfloor \tilde{\mathbf{x}}^* \rfloor\} \\
      && \mathcal{N}^-_Q=\{i \in \mathcal N \mid \phi_i \in Q, i\notin \lfloor \tilde{\mathbf{x}}^* \rfloor\} \\
      && c_{max}(\mathcal{N}^+_Q) = \max_{i \in \mathcal{N}^+_Q} c_i, ~c_{min}(\mathcal{N}^+_Q) = \min_{i \in \mathcal{N}^+_Q} c_i \\
      && c_{max}(\mathcal{N}^-_Q) = \max_{i \in \mathcal{N}^+_Q} c_i, ~c_{min}(\mathcal{N}^-_Q) = \min_{i \in \mathcal{N}^+_Q} c_i \\
      %
      %
      %
      && \delta_{min} = \min_{Q, Q' \in \Omega} \left| \frac{\mu_{Q}}{c_{min}(\mathcal{N}^-_Q)} - \frac{\mu_{Q'}}{c_{max}(\mathcal{N}^+_Q)}   \right| \\
      && \xi = \frac{8}{c^2_{min}\delta^2_{min}} + \left( \frac{c_{max}}{c_{min}} \right)^2
    \end{eqnarray}

  \subsection{Main Result} \label{ssec:main}
    As shown in Sec.~\ref{sec:sys}, the context space $\mathcal S$ is partitioned according to the granularity $d$. Increasing $d$ results in more fine-grained hypercubes such that we can estimate the workers' sensing abilities more accurately (as shown in \textbf{Lemma}~\ref{le:errincube}). Nevertheless, increasing the granularity also implies we have more hypercubes to explore and exploit, while the exploration and exploitation are restricted by the limited total budget. To this end, as shown in the following \textbf{Theorem}~\ref{thm:main}, we fine tune the granularity of the partitioning such that the regret function can be properly bounded.
    \begin{theorem}  \label{thm:main}
      Assuming $d = \left\lceil B^{\frac{1}{\alpha+M}} \right\rceil$, the regret function of our CAWS algorithm (\ref{eq:reg}) is upper-bounded by
      %
      %
      %
      \begin{equation}
        \left( \tau_{max} + 2^M B^{\frac{M}{\alpha+M}}h(\ln B) +1 \right)\frac{c_{max}}{c_{min}} + \frac{4L M^{\frac{\alpha}{2}} B^{\frac{M}{\alpha+M}}}{c_{min}} + 1
      \end{equation}
      where
      \begin{equation} \label{eq:hlnb}
        h(\ln B) = \xi \ln \left( \frac{B}{c_{min}} \right) + \frac{\pi^2}{3} + 1
      \end{equation}
      which implies that the regret of our CAWS algorithm is $\mathcal{O} \left( B^{\frac{M}{\alpha+M}} \ln B \right)$.
    %

    \end{theorem}

  \subsection{Detailed Proof} \label{ssec:proof}
  %
    %
    \begin{theorem} \label{thm:reg1bd}
      Recall that $\lfloor \mathbf{x}^* \rfloor = \{\lfloor x^*_i \rfloor\}^N_{i=1}$ and $\lfloor \tilde{\mathbf{x}}^* \rfloor = \{\lfloor \tilde x^*_i \rfloor\}^N_{i=1}$ are the results we obtain by applying the rounding-based density-ordered greedy algorithm to the two BKP instances $\mathsf{Instance1}=(\{i, \mu_i, c_i, \tau_i\}^N_{i=1}, B)$ and $\mathsf{Instance2}=(\{i, \mu_{Q_i}, c_i, \tau_i\}^N_{i=1}, B)$, respectively. Considering $\mu_{Q_i}$ is an estimate on $\mu_i$ for $\forall i \in \mathcal N$, we have
      \begin{equation} \label{eq:reg1bd}
        \mathsf{Regret}(\lfloor \mathbf{x}^* \rfloor, \lfloor \tilde{\mathbf x}^* \rfloor, \{\mu_i\}^N_{i=1}) \leq \frac{2 \Delta B}{c_{min}}+1
      \end{equation}
    \end{theorem}
    \begin{proof}
      We denote by $\mathbf{x}^*=\{x^*_i\}^N_{i=1}$ and $\tilde{\mathbf x}^*=\{\tilde x^*_i\}^N_{i=1}$ the fractional solutions to the FBKP versions of $\mathsf{Instance1}$ and $\mathsf{Instance2}$, respectively. Considering the inequality (\ref{eq:optsoluinq}),
      \begin{align} \label{eq:reg1ineq}
        \mathsf{Regret}(\lfloor \mathbf{x^*} \rfloor, \lfloor \tilde{\mathbf x}^* \rfloor, \{\mu_i\}^N_{i=1}) \leq \sum^N_{i=1} \mu_i x^*_i - \left( \sum^N_{i=1}\mu_i \tilde{x}^*_i - \mu_{Q_{\tilde k}} \right) \leq \sum^N_{i=1} \mu_i x^*_i - \sum^N_{i=1}\mu_i \tilde{x}^*_i + 1
      \end{align}
      where $\tilde k$ is the split worker in $\mathsf{Instance2}$ and $\mu_{\tilde k} \leq 1$. According to the procedure of our rounding-based density-ordered greedy algorithm shown in Sec.~\ref{sec:analysis}, if there is a worker $i$ such that $x^*_i > \tilde{x}^*_i$, there must be at least one another worker $i'$ with $ x^*_{i'}  <  \tilde{x}^*_{i'} $ such that $\frac{\mu_{i'}}{c_{i'}} \leq \frac{\mu_{i}}{c_{i}}$ and $\frac{\mu_{Q_{i'}}}{c_{i'}} \geq \frac{\mu_{Q_{i}}}{c_{i}}$. Therefore,
      \begin{align} \label{reg1ineq2}
        \frac{\mu_i}{c_i} - \frac{\mu_{i'}}{c_{i'}} \leq \frac{\mu_{Q_i}+\Delta}{c_i} - \frac{\mu_{Q_{i'}}-\Delta}{c_{i'}} = \frac{\mu_{Q_i}}{c_i} - \frac{\mu_{Q_{i'}}}{c_{i'}} + \Delta\left(\frac{1}{c_i}+\frac{1}{c_{i'}}\right) \leq \frac{2\Delta}{c_{min}}
      \end{align}
      where we have the first inequality due to $|\mu_i - \mu_{Q_i}| \leq \Delta$ for $\forall i\in \mathcal N$ (see \textbf{Lemma}~\ref{le:errincube}) and the third one by considering the facts that $\frac{\mu_{Q_{i'}}}{c_{i'}} \geq \frac{\mu_{Q_{i}}}{c_{i}}$ and $c_i \geq c_{min}$ for $\forall i \in \mathcal N$. 
      In other words, if worker $i$ is not (fractionally) selected in $\mathsf{Instance2}$, our algorithm will select some another workers to replace worker $i$ in $\widetilde{\mathbf{x}}^*$. Nevertheless, since the workers (selected to replace worker $i$) have smaller densities than $i$, these replacements may result in reward loss, which can be indicated by the difference between the first two terms in (\ref{eq:reg1ineq}). Hence, assuming $\widetilde{\mathcal N}_i$ denote the set of those workers and $\tilde i = \arg\min_{j \in \widetilde{\mathcal N}_i} {\mu_{Q_j}}/{c_j}$, we have
      \begin{eqnarray}
        && \sum^N_{i=1} \mu_i x^*_i - \sum^N_{i=1}\mu_i \tilde{x}^*_i \nonumber\\
        &\leq& \sum_{i: x^*_i > \tilde{x}^*_i}\left( \left( x^*_i - \tilde{x}^*_i \right)\mu_i - \frac{(x^*_i - \tilde{x}^*_i)c_i}{c_{\tilde i}}\cdot \mu_{\tilde i} \right) \nonumber\\
        &=& \sum_{i: x^*_i > \tilde{x}^*_i}\left( c_i(x^*_i - \tilde x^*_i)\left( \frac{\mu_i}{c_i} - \frac{\mu_{\tilde i}}{c_{\tilde i}} \right) \right)
      \end{eqnarray}
      Considering $\frac{\mu_i}{c_i} - \frac{\mu_{\tilde i}}{c_{\tilde i}} \leq \frac{2\Delta}{c_{min}}$ (see (\ref{reg1ineq2})), we have
      \begin{eqnarray*}
        \sum^N_{i=1} \mu_i x^*_i - \sum^N_{i=1}\mu_i \tilde{x}^*_i &\leq& \frac{2\Delta}{c_{min}} \sum_{i: x^*_i > \tilde{x}^*_i} c_i(x^*_i - \tilde x^*_i) \leq \frac{2 \Delta B}{c_{min}}
      \end{eqnarray*}
      by substituting which into (\ref{eq:reg1ineq}), we complete the proof.
    \end{proof}

    %
    %
    \begin{theorem} \label{thm:reg2upbd}
      Letting $\lfloor \tilde{\mathbf{x}}^* \rfloor$ be the results obtained by applying the rounding-based density-ordered greedy algorithm to the BKP instance $(\{i, \mu_{Q_i}, c_i, \tau_i\}^N_{i=1}, B)$ and $\mathbf{x}$ be the output of our CAWS algorithm, we have
      \begin{align} \label{eq:reg2upbd}
        \mathsf{Regret}\left( \lfloor \tilde{\mathbf x}^* \rfloor, \mathbf x, \{\mu_{Q_i}\}^N_{i=1} \right) \leq \frac{c_{max}}{c_{min}} \left( \tau_{max} + d^M h(\ln B)+1 \right)
      \end{align}
    \end{theorem}
    \begin{proof}
      We introduce a redundant term $\mathbb E_T \left[ T \right] \mu_{Q_{i^*}}$ such that
      \begin{align} \label{eq:reg2decom}
        \mathsf{Regret}\left( \lfloor \tilde{\mathbf x}^* \rfloor, \mathbf x, \{\mu_{Q_i}\}^N_{i=1} \right) = \sum^N_{j=1} \mu_{Q_j} \lfloor \tilde{x}^*_i \rfloor - \mu_{Q_{i^*}} \mathbb E_T \left[ T \right] + \mathbb E_T \left[ T\mu_{Q_{i^*}} - \sum_{Q \in \Omega} \sum_{i \in \mathcal N_{Q}} \mu_Q \mathbb E_{\{i(t)\}^T_{t=1}}[x_i \mid T] \right]
      \end{align}
      where the lower-bound of the second term and the upper-bound of the third one are given in \textbf{Lemma}~\ref{le:iterations}. The proof of \textbf{Lemma}~\ref{le:iterations} can be found in \textbf{Appendix A}
      \begin{lemma} \label{le:iterations}
        Supposing $T$ denotes the total number of the iterations our CAWS proceeds with budget $B$, we have the following two inequalities holds
        %
        \begin{align} \label{eq:bdt}
          \mathbb E_T[T] \geq& \frac{B-c_{max}}{c_{i^*}} - \mathbb E_T \left[ \sum_{j \in \lfloor \tilde{\mathbf x}^* \rfloor} \frac{c_j - c_{i^*}}{c_{i^*}} \mathbb E_{\{i(t)\}^T_{t=1}} [x_j \mid T]  \right] - \sum_{Q : c_{max}(\mathcal N^-_Q) > c_{i^*}} \hspace{-5ex} \frac{c_{max}(\mathcal N^-_Q) - c_{i^*}}{c_{i^*}} h(\ln B)
        \end{align}
        and
        \begin{align} \label{eq:reg22}
          &\mathbb E_T \left[ T\mu_{Q_{i^*}} - \sum_{Q \in \Omega} \sum_{i \in \mathcal N_{Q}} \mu_Q \mathbb E_{\{i(t)\}^T_{t=1}}[x_i\mid T] \right]  \nonumber\\
          \leq& \sum_{Q: \mu_{Q_{i^*}} > \mu_Q} (\mu_{Q_{i^*}} - \mu_Q) h(\ln B)   -  \mathbb E_T \left[ \sum_{j \in \lfloor \tilde{\mathbf x}^* \rfloor} (\mu_{Q_{i^*}}-\mu_{Q_j}) \mathbb E_{\{i(t)\}^T_{t=1}}[x_j \mid T] \right]  
        \end{align}
      \end{lemma}

      By substituting the above two inequalities (\ref{eq:bdt}) and (\ref{eq:reg22}) into (\ref{eq:reg2decom}), we have
      \begin{align} \label{eq:reg2upbd1}
        & \mathsf{Regret}\left( \lfloor \tilde{\mathbf x}^* \rfloor, \mathbf x, \{\mu_{Q_i}\}^N_{i=1} \right) \nonumber\\
        %
        %
        %
        =& \sum_{Q \in \Omega}\sum_{j \in Q} \mu_Q \lfloor \tilde{x}^*_j \rfloor - \mathbb E_T [ T ]\mu_{Q_{i^*}} + \mathbb E_T \left[ T\mu_{Q_{i^*}} - \sum_{Q \in \Omega} \sum_{i \in \mathcal N_{Q}} \mu_Q \mathbb E_{\{i(t)\}^T_{t=1}}[x_i \mid T] \right] \nonumber\\
        \leq& \sum_{Q \in \Omega}\sum_{j \in Q} \mu_Q \lfloor \tilde{x}^*_j \rfloor - \frac{\mu_{Q_{i^*}}(B-c_{max})}{c_{i^*}} + \mu_{Q_{i^*}}\mathbb E_T \left[ \sum_{j \in \lfloor \tilde{\mathbf x}^* \rfloor} \frac{c_j - c_{i^*}}{c_{i^*}} \mathbb E_{\{i(t)\}^T_{t=1}} [x_j \mid T]  \right] \nonumber\\
        & + \mu_{Q_{i^*}} \sum_{Q: c_{max}(\mathcal N^-_Q) > c_{i^*}} \hspace{-4ex} \frac{c_{max}(\mathcal N^-_Q) - c_{i^*}}{c_{i^*}} h(\ln B)  + \sum_{Q: \mu_{Q_{i^*}} > \mu_Q} (\mu_{Q_{i^*}} - \mu_Q) h(\ln B) \nonumber\\
        & +  \mathbb E_T \left[ \sum_{j \in \lfloor \tilde{\mathbf x}^* \rfloor} (\mu_{Q_{i^*}}-\mu_{Q_j}) \mathbb E_{\{i(t)\}^T_{t=1}}[x_j \mid T] \right]  \nonumber\\
        \leq& \sum_{Q \in \Omega}\sum_{j \in Q} \mu_Q \lfloor \tilde{x}^*_j \rfloor - \frac{B\mu_{Q_{i^*}}}{c_{i^*}} + \frac{c_{max}\mu_{Q_{i^*}}}{c_{i^*}} \nonumber\\
        &+  \mathbb E_T \Bigg[ \sum_{j \in \lfloor \tilde{\mathbf x}^* \rfloor} \left( \frac{\mu_{Q_{i^*}}c_j}{c_{i^*}} - \mu_{Q_j} \right) E_{\{i(t)\}^T_{t=1}}[x_j \mid T] \Bigg] + \sum_{Q\in \Omega} g \cdot h(\ln B)
      \end{align}
      where 
      \begin{eqnarray}
        g = \mathbb I(c_{max}(\mathcal N^-_Q) > c_{i^*}) \cdot \frac{\mu_{Q_{i^*}}(c_{max}(\mathcal N^-_Q) - c_{i^*})}{c_{i^*}} + \mathbb I(\mu_{Q_{i^*}} - \mu_Q>0) \cdot (\mu_{Q_{i^*}} - \mu_Q)
      \end{eqnarray}
      Since $B \geq \sum_{j\in \lfloor \tilde{\mathbf{x}}^* \rfloor} c_j \lfloor \tilde{x}^*_j \rfloor$, we have
      %
      \begin{eqnarray} \label{eq:reg2upbd2}
        &&\sum_{Q \in \Omega}\sum_{j \in \Omega} \mu_Q \lfloor \tilde{x}^*_j \rfloor - \frac{B\mu_{Q_{i^*}}}{c_{i^*}} \nonumber\\
        &\leq& \sum_{j \in \lfloor \tilde{\mathbf x}^* \rfloor } \mu_{Q_j} \lfloor \tilde{x}^*_j \rfloor - \frac{\mu_{Q_{i^*}}\sum_{j\in \lfloor \tilde{\mathbf{x}}^* \rfloor} c_j \lfloor \tilde{x}^*_j \rfloor}{c_{i^*}} \nonumber\\
        &=& \sum_{j \in \lfloor \tilde{\mathbf x}^* \rfloor } \left( \mu_{Q_j} - \frac{\mu_{Q_{i^*}}c_j}{c_{i^*}} \right) \lfloor \tilde{x}^*_j \rfloor
      \end{eqnarray}
      In addition, since $0 \leq \mu_{Q_{i^*}} \leq 1$, $c_{max}(\mathcal N^-_Q) - c_{i^*} \leq c_{max} - c_{min}$ and $\mu_{Q_{i^*}} - \mu_{Q} \leq 1$, we have
      \begin{eqnarray} \label{eq:reg2upbd3}
        g \leq \frac{c_{max}-c_{min}}{c_{min}} + 1 = \frac{c_{max}}{c_{min}}
      \end{eqnarray}
      Substituting (\ref{eq:reg2upbd2}) and (\ref{eq:reg2upbd3}) into (\ref{eq:reg2upbd1}), we have
      %
      %
      \begin{align} \label{eq:reg2upbd4}
        &\mathsf{Regret}\left( \lfloor \tilde{\mathbf x}^* \rfloor, \mathbf x, \{\mu_{Q_i}\}^N_{i=1} \right) \nonumber\\
        \leq& \sum_{j \in \lfloor \tilde{\mathbf x}^* \rfloor } \left( \mu_{Q_j} - \frac{\mu_{Q_{i^*}}c_j}{c_{i^*}} \right) \lfloor \tilde{x}^*_j \rfloor + \frac{c_{max}\mu_{Q_{i^*}}}{c_{i^*}} + d^M \frac{c_{max}}{c_{min}} \left( \xi \ln T + \frac{\pi^2}{3} + 1 \right) \nonumber\\
        & + \mathbb E_T \Bigg[ \scriptstyle{ \sum_{j \in \lfloor \tilde{\mathbf x}^* \rfloor} \left( \frac{\mu_{Q_{i^*}}(c_j - c_{i^*})}{c_{i^*}}+(\mu_{Q_{i^*}}-\mu_{Q_j}) \right) \mathbb E_{\{i(t)\}^T_{t=1}}[x_j \mid T] } \Bigg] \nonumber\\
        =& \sum_{j \in \lfloor \tilde{\mathbf x}^* \rfloor } \left( \mu_{Q_j} - \frac{\mu_{Q_{i^*}}c_j}{c_{i^*}} \right) \lfloor \tilde{x}^*_j \rfloor + \frac{c_{max}\mu_{Q_{i^*}}}{c_{i^*}} + d^M \frac{c_{max}}{c_{min}} \left( \xi \ln \left( \frac{B}{c_{min}} \right) + \frac{\pi^2}{3} + 1 \right) \nonumber\\
        & -  \mathbb E_T \left[ \sum_{j \in \lfloor \tilde{\mathbf x}^* \rfloor} \left( \mu_{Q_j} - \frac{\mu_{Q_{i^*}}c_j}{c_{i^*}} \right) \mathbb E_{\{i(t)\}^T_{t=1}}[x_j \mid T] \right] \nonumber\\
        \leq& \mathbb E_T \Bigg[ \sum_{j \in \lfloor \tilde{\mathbf x}^* \rfloor} \frac{\mu_{Q_{i^*}}c_j}{c_{i^*}}  \left( \mathbb E_{\{i(t)\}^T_{t=1}}[x_j \mid T] - \lfloor \tilde{x}^*_j \rfloor  \right) \Bigg]   d^M \frac{c_{max}}{c_{min}} \left( \xi \ln \left( \frac{B}{c_{min}} \right) + \frac{\pi^2}{3} + 1 \right) + \frac{c_{max}\mu_{Q_{i^*}}}{c_{i^*}} 
      \end{align}
      where we have the last inequality holds by considering $\frac{\mu_{Q_{i^*}}}{c_{i^*}} \geq \frac{\mu_{Q_{j}}}{c_{j}}$ for $\forall j\in \mathcal N$.

      The first term at the right side of the above inequality can be written as
      \begin{align} \label{eq:reg2upbd5}
        & \quad\quad \mathbb E_T \left[ \sum_{j \in \lfloor \tilde{\mathbf x}^* \rfloor}  \frac{\mu_{Q_{i^*}}c_j}{c_{i^*}} \left( \mathbb E_{\{i(t)\}^T_{t=1}}[x_j \mid T] - \lfloor \tilde{x}^*_j \rfloor  \right) \right] \nonumber\\
        &= \sum_{j \in \lfloor \tilde{\mathbf x}^* \rfloor} \frac{\mu_{Q_{i^*}}c_j}{c_{i^*}}  \left( \mathbb E_T [\mathbb E_{\{i(t)\}^T_{t=1}}[x_j \mid T]] - \lfloor \tilde{x}^*_j \rfloor  \right) \nonumber\\
        &= \sum_{j \in \lfloor \tilde{\mathbf x}^* \rfloor} \frac{\mu_{Q_{i^*}}c_j}{c_{i^*}}  \left( \mathbb E_T [\mathbb E_{\{i(t)\}^T_{t=1}}[x_j - \lfloor \tilde{x}^*_j \rfloor \mid T]]   \right) \nonumber\\
        &\leq \hspace{-2ex} \sum_{j \in \lfloor \tilde{\mathbf x}^* \rfloor : x_j > \lfloor \tilde{x}^*_j \rfloor} \frac{\mu_{Q_{i^*}}c_j}{c_{i^*}} \left( \mathbb E_T [\mathbb E_{\{i(t)\}^T_{t=1}}[x_j - \lfloor \tilde{x}^*_j \rfloor \mid T]]   \right)
      \end{align} 
      Suppose $\tilde k$ denotes the split worker in $\lfloor \tilde{\mathbf x}^* \rfloor$. For any worker $j$ such that $j\in \lfloor \tilde{\mathbf x}^* \rfloor$ and $j \neq \tilde k$, we have $\lfloor \tilde{x}^*_j \rfloor = \tau_j$, while $\lfloor \tilde{x}^*_{\tilde k} \rfloor \leq \tau_{\tilde k}$. Therefore, for $\forall j\in \lfloor \tilde{\mathbf x}^* \rfloor$, $x_j - \lfloor \tilde{x}^*_j \rfloor \leq 0$, and the split worker $\tilde k$ is the only possible one such that $x_{\tilde k} - \lfloor \tilde{x}^*_{\tilde k} \rfloor \geq 0$ may hold. Also, since $c_j \leq c_{max}$ for $\forall j\in \mathcal N$, continuing the above equation (\ref{eq:reg2upbd5}), we have
      \begin{align} \label{eq:reg2upbd6}
        \mathbb E_T \left[ \sum_{j \in \lfloor \tilde{\mathbf x}^* \rfloor}  \frac{\mu_{Q_{i^*}}c_j}{c_{i^*}}  \left( \mathbb E_{\{i(t)\}^T_{t=1}}[x_j \mid T] - \lfloor \tilde{x}^*_j \rfloor  \right) \right] \leq \frac{\tau_{max}\mu_{Q_{i^*}}c_{max}}{c_{i^*}}
      \end{align} 
      We complete the proof by substituting (\ref{eq:reg2upbd6}) into (\ref{eq:reg2upbd4}) as follows
      \begin{align} \label{eq:reg2upbd7}
        & \mathsf{Regret}\left( \lfloor \tilde{\mathbf x}^* \rfloor, \mathbf x, \{\mu_{Q_i}\}^N_{i=1} \right) \nonumber\\
        \leq& \frac{\tau_{max}\mu_{Q_{i^*}}c_{max}}{c_{i^*}}  + \frac{c_{max}\mu_{Q_{i^*}}}{c_{i^*}} + d^M \frac{c_{max}}{c_{min}} h(\ln B) \nonumber\\
        \leq& (\tau_{max}+1)\frac{c_{max}}{c_{min}} + d^M \frac{c_{max}}{c_{min}}h(\ln B)
      \end{align}
    \end{proof}

    Now, we are ready to prove our main result shown in \textbf{Theorem}~\ref{thm:main}. Combining \textbf{Theorem}~\ref{thm:reg1bd} and \textbf{Theorem}~\ref{thm:reg2upbd} into (\ref{eq:regdecom}), we have
    \begin{align} \label{eq:main2}
      \mathsf{Regret}(T, \{i(t)\}^T_{t=1}) \leq \left( \tau_{max} + d^M h(\ln B) + 1 \right)\frac{c_{max}}{c_{min}} + \frac{4 \Delta B}{c_{min}} + 1 
      %
    \end{align}
    Letting $d = \left\lceil B^{\frac{1}{\alpha+M}} \right\rceil$, it follows
    \begin{equation} \label{eq:dmupbd}
      d^M = {\left\lceil B^{\frac{1}{\alpha+M}} \right\rceil}^M \leq 2^M B^{\frac{M}{\alpha+M}}
    \end{equation}
    and
    \begin{equation} \label{eq:deltaupbd}
      \Delta =L \left( M^{\frac{1}{2}}d^{-1} \right)^\alpha \leq L M^{\frac{\alpha}{2}} B^{-\frac{\alpha}{\alpha+M}}
    \end{equation}
    when $\alpha > 0$ as shown in \textbf{Assumption}~\ref{ap:holder}. We finally complete the proof by substituting (\ref{eq:dmupbd}) and (\ref{eq:deltaupbd}) into (\ref{eq:main2}).

    \subsection{Discussion} \label{ssec:dis}
      As shown above, the upper-bound on the regret can be written as $\mathcal O(d^M \ln B)$, which can be further represented as $\mathcal O(B^{\frac{M}{\alpha+M}}\ln B)$ by substituting $d = \left\lceil B^{\frac{1}{\alpha+M}} \right\rceil$. It is illustrated that the granularity of partitioning, i.e., $d$, is one of the key factors dominating the upper-bound. Although there have been many existing proposals investigating the budget-limited worker selection problem, they usually exploit and explore the individual workers directly (see Sec.~\ref{sec:relwork} later). For example, in \cite{TranCRJ-AAAI12,HanZL-TON16,RangiF-AAMAS18}, if taking into account the number of workers (i.e., $K$), individually learning the workers' sensing abilities results in an upper-bound $\mathcal O(K \ln B)$ on the regret. Fortunately, the exploration-exploitation trade-off in our CAWS algorithm is made in the context space, and we fine tune the granularity of partitioning such that $d^M \ll K$ especially when budget $B$ is considerably limited while $K$ is rather huge. Hence, our CAWS algorithm is of significant scalability in selecting among massive unknown workers with limited budget. Furthermore, due to the budget constraint, the number of iterations $T$ is $\Theta(B)$, since $\frac{B}{c_{max}} \leq T \leq \frac{B}{c_{min}}$. Therefore, the upper-bound on the regret can be re-written as $\mathcal O(d^M \ln T)$ where $d^M$ is the number of the hypercubes we explore and exploit in our CAWS algorithm. Considering the upper-bound is $\mathcal O(K\ln T)$ for the regret of the standard MAB problem where no constraint on the total budget is considered and $K$ represents the number of arms (corresponding to workers in our case), the advantage of our CAWS algorithm can be further confirmed. We will also verify the efficacy of our algorithm by extensive experiments on both synthetic and real datasets later in Sec.~\ref{sec:exp}.
      
      Similar to most of existing proposals (e.g., \cite{HanZL-TON16,RangiF-AAMAS18,GaoWYXC-ICPADS19,GaoWXC-INFOCOM20}), we currently assume the workers have static sensing abilities, since each worker has fixed context in our case; nevertheless, the context for a worker may be time-varying in some application scenarios such that the workers may have their sensing abilities changed over time. For example, a mobile worker usually have distinct sensing abilities at different locations which are far apart from each other, when conducting a location-based sensing task. Since our CAWS algorithm aims at learning the correlation between context and sensing ability rather than the sensing abilities of individual workers, it can estimate the sensing abilities of the workers according to their instant contexts, instead of re-learning once their sensing abilities are changed. Specifically, we let each worker $i$ first report its instant context $\phi_i(t)$ to the crowdsensing platform in each iteration $t$, and the platform then can make a proper selection decision by taking $\{\phi_i(t)\}^N_{i=1}$ as the input parameters of \textbf{Algorithm}~\ref{alg:contextalgo}. We will verify the efficacy of this adaptation by extensive experiments on a dataset of vehicular trajectories later in Sec.~\ref{sec:exp}.

\section{Experiments} \label{sec:exp}
  In this section, we evaluate the performance of our CAWS algorithm in terms of selecting among massive unknown workers under a limited budget through extensive experiments. We first introduce the reference algorithms in Sec.~\ref{ssec:baseline} and then compare them with our CAWS algorithm using a synthetic dataset and two real datasets in Sec.~\ref{ssec:synthetic} and Sec.~\ref{ssec:real}, respectively. We finally report the experiments in adapting our CAWS algorithm to select among massive unknown mobile workers with time-varying contexts in Sec.~\ref{ssec:ext}.

  \subsection{Reference Algorithms} \label{ssec:baseline}
    We mainly compare our CAWS algorithm with the following ones which can be applied to our problem.
    \begin{itemize}
      \item \textbf{Oracle}: Oracle is aware of the sensing ability of each worker; therefore, it applies the density-ordered greedy algorithm to output a nearly optimal solution.
      \item \textbf{Bounded $\epsilon$-first}: The bounded $\epsilon$-first algorithm is with decoupled exploitation and exploration \cite{ThanhSRJ-AI14}. Under a $\epsilon$-fraction of the budget, it explores the workers uniformly to estimate their sensing abilities; while with the remaining budget, it assigns the task to the workers according to their estimated sensing abilities through the density-ordered greedy algorithm. 
      \item \textbf{B-KUBE}: B-KUBE is a CMAB-based algorithm to handle BKPs, where workers with unknown sensing abilities and bounded sensing capacities are selected under a given budget \cite{RangiF-AAMAS18}. It can be considered to be a degeneration of CAWS where the context space is sufficiently partitioned such that each hypercube contains only one worker. Our CAWS algorithm is then degraded to that we estimate the workers' sensing abilities directly by their UCB indices which are calculated according to their historical performances. We have conducted theoretical comparison between it and our CAWS algorithm in Sec.~\ref{ssec:dis}, and we now continue the comparison between the two algorithms through numerical experiments.
      \item \textbf{Random}: The (purely) random algorithm selects an available worker (whose residential capacity is non-zero) uniformly in each iteration until the budget is exhausted or none of the workers have non-zero residual capacities.
    \end{itemize}

  \subsection{Evaluation with Synthetic Data} \label{ssec:synthetic}
    We hereby quantitatively evaluate the above algorithms based on synthetic data. We conduct our experiments by assuming there are $10^5$ workers whose capacities and costs are distributed uniformly in $[20,40]$ and $[1, 1.5]$, respectively. We suppose the context space $\mathcal S$ has $M=2$ dimensions and each dimension is normalized to $[0,1]$ as mentioned in Sec.~\ref{ssec:model}. The workers have their contexts uniformly distributed in the context space $\mathcal S$. We also randomly set the workers' sensing abilities such that the H$\ddot{\text{o}}$lder condition holds for $\alpha=1$ for the purpose of quantitative analysis (e.g., one choice is to let each worker have its sensing ability being the numerical average of its context dimensions).

    We first study the performance of the different algorithms in terms of cumulative sensing revenue (i.e., the total actual reward obtained in practice) and report the results in Fig.~\ref{fig:synbudget}~(a). We vary the budget from $4 \times 10^4$ to $4 \times 10^5$ with a step size $4 \times 10^4$. Note that the budget is only at most $4$ times higher (or even smaller) than the number of the workers in this setting. In a nutshell, compared with $N$, $B$ is quite limited. It is shown in Fig.~\ref{fig:synbudget}~(a) that our CAWS algorithm yields higher cumulative revenue than the others, since the context information of the workers can be effectively utilized in CAWS algorithm to estimate the workers' sensing abilities even we do not have sufficient budget to enable a direct estimation for each worker. Furthermore, the performance of our algorithm is very close to the one of the oracle algorithm, especially when the budget is more limited. We also plot the regrets of the different algorithms in Fig.~\ref{fig:synbudget}~(b). Since the regret of the oracle algorithm is always almost zero, we do not show it in Fig.~\ref{fig:synbudget}~(b). Consistent with what has been shown in Fig.~\ref{fig:synbudget}~(a), our CAWS algorithm has a much lower regret than the other three alternatives. When the budget is increased, our algorithm proceeds more iterations such that the regret is increased but at a very low rate, which is consistent with our theoretical result in \textbf{Theorem}~\ref{thm:main}. 
    \begin{figure}[htb!]
    \begin{center}
      \parbox{.49\textwidth}{\center\includegraphics[width=.46\textwidth]{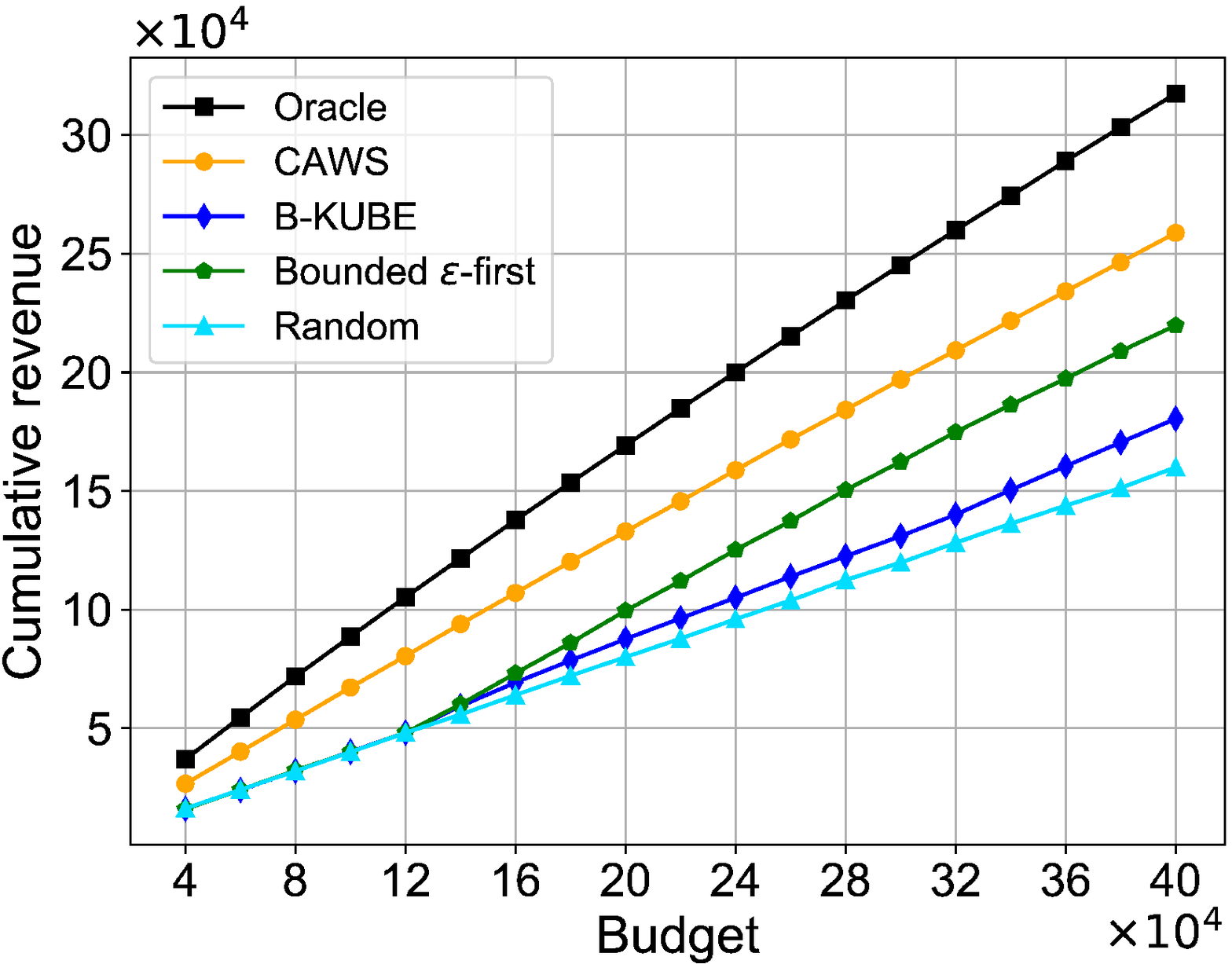}}
      \parbox{.49\textwidth}{\center\includegraphics[width=.48\textwidth]{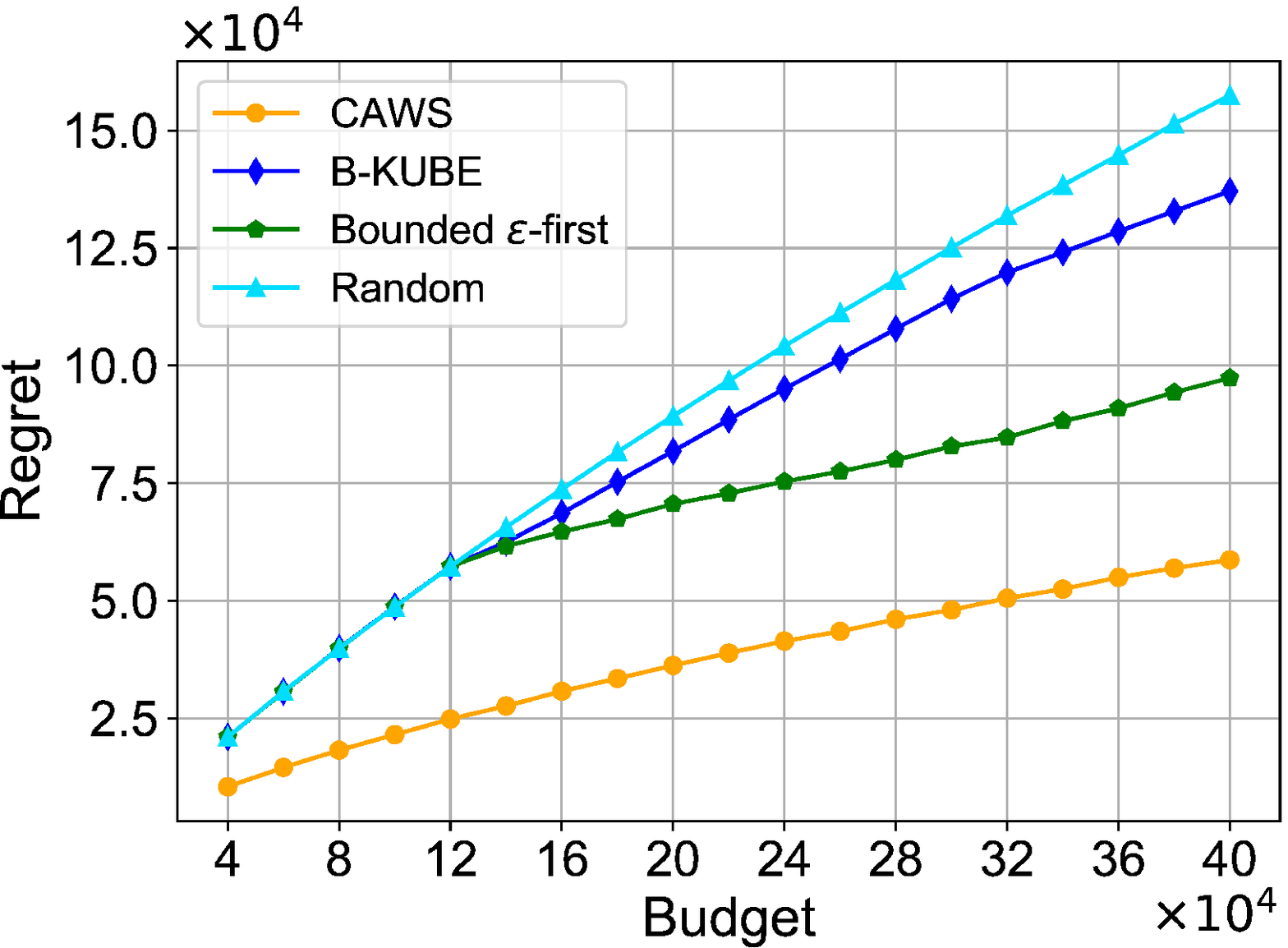}}
      \parbox{.49\textwidth}{\center\scriptsize(a) Cumulative revenue}
      \parbox{.49\textwidth}{\center\scriptsize(b) Regret}
      \caption{Comparisons with different budgets on synthetic dataset.}
      \label{fig:synbudget}
    \end{center}
    \end{figure}

    We then fix budget $B=1 \times 10^5$ and vary the number of workers $N=4,6,8,10 \times 10^4$ to show the scalabilities of the different algorithms in face of an increasing number of workers. The results in terms of cumulative revenue and regret are presented in Fig.~\ref{fig:syndiffn}. It is demonstrated that, regardless of how many workers are given, our CAWS algorithm can yield more cumulative revenue and result in much smaller regret than the other reference algorithms. Especially when there are $10^5$ workers, for each of the reference algorithms, its regret is even two times higher than the one of our CAWS algorithm. Moreover, even when the number of workers is increased, the cumulative revenue and the regret of our CAWS algorithm are stable, since the performance of our algorithm mainly depends on the budget as shown in \textbf{Theorem}~\ref{thm:main} while the budget is fixed in our setting. As there may be more elite workers participating in the sensing task when the total number of the workers is increased, both the oracle algorithm and our CAWS algorithm yield a little more cumulative revenue. Additionally, due to the exploration-exploitation trade-off in our CAWS algorithm, there is a very slight increase in the regret of our algorithm. In contrast, the performances of the other algorithms are obviously degraded in face of a large number of unknown workers, as they have no sufficient budget to explore and exploit the workers individually.
    \begin{figure}[htb!]
    \begin{center}
      \parbox{.49\textwidth}{\center\includegraphics[width=.48\textwidth]{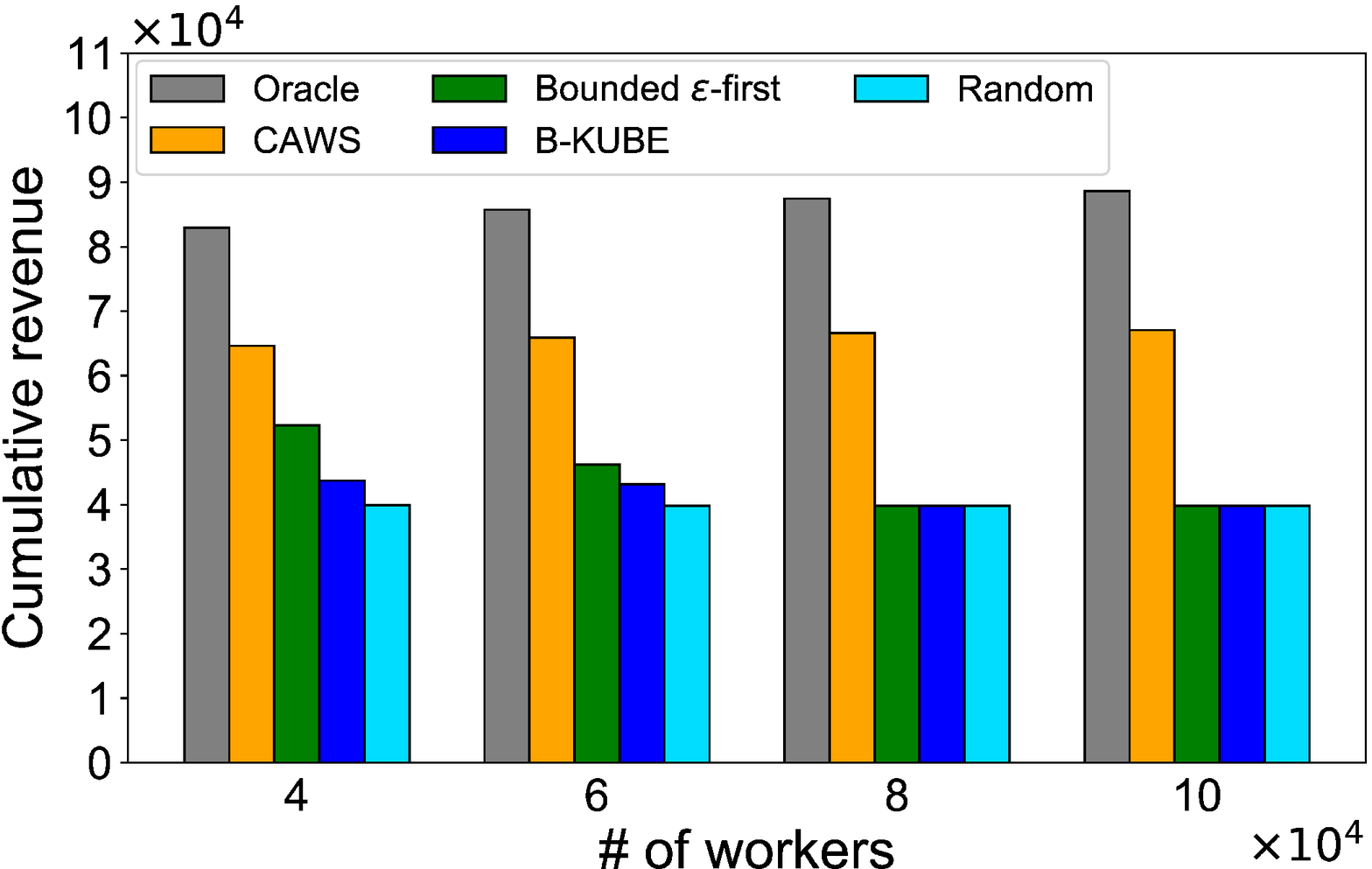}}
      \parbox{.49\textwidth}{\center\includegraphics[width=.48\textwidth]{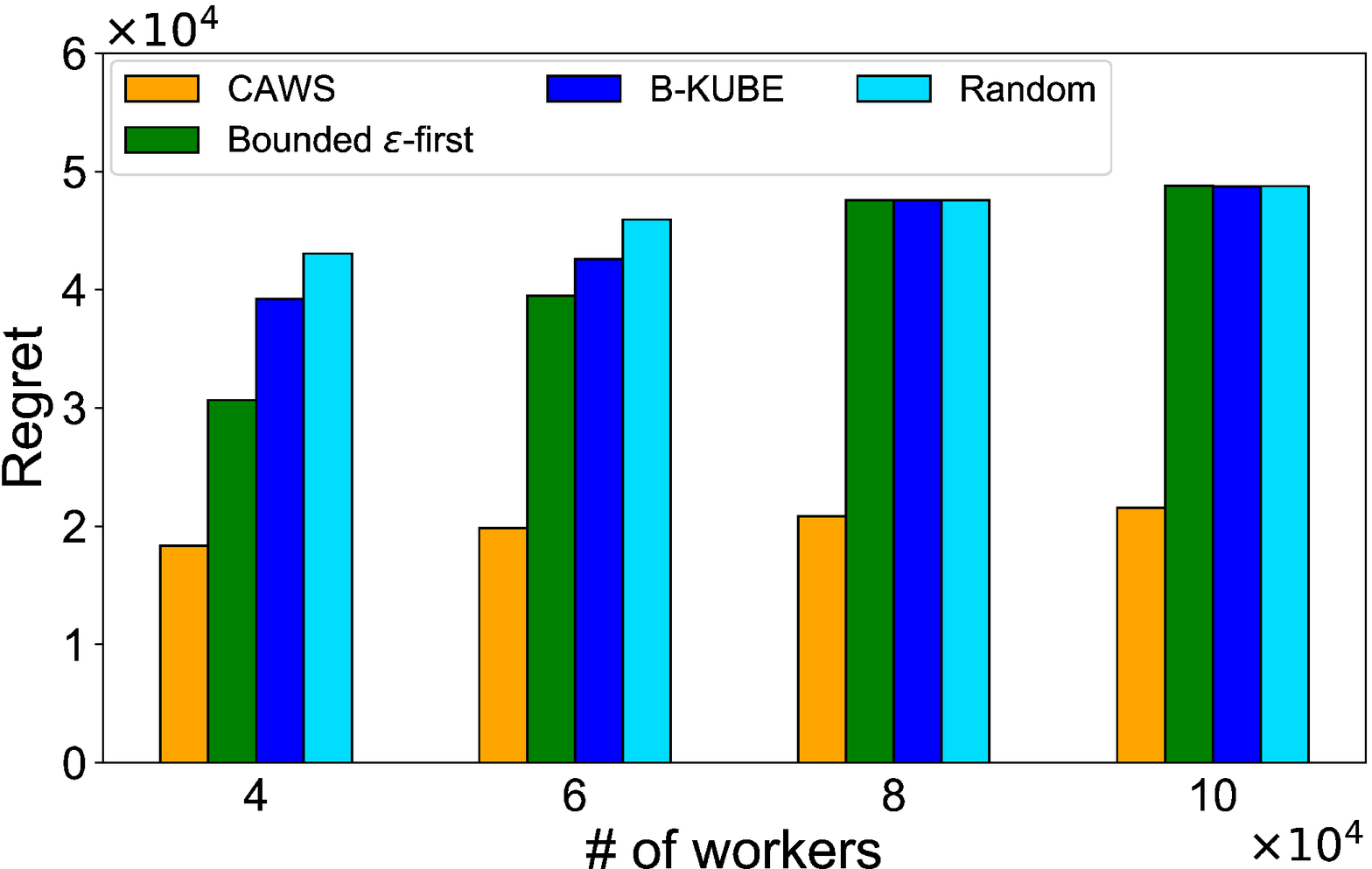}}
      \parbox{.49\textwidth}{\center\scriptsize(a) Cumulative revenue}
      \parbox{.49\textwidth}{\center\scriptsize(b) Regret}
      \caption{Comparisons with different numbers of workers on synthetic dataset.}
      \label{fig:syndiffn}
    \end{center}
    \end{figure}
  
  \subsection{Evaluation with Real Data} \label{ssec:real}
    \subsubsection{Experiments on Vehicular Trajectory Dataset} \label{sssec:szstatic}
      We first adopt a vehicular trajectory dataset consisting of $13,798$ taxicabs' GPS coordinates collected over $8$ days in Shenzhen, China~\cite{YangWYJ-TITS19,LiKWJ-TITS21}. Each data sample in the dataset contains a vehicle’s license plate number, longitude location, latitude location, etc. We randomly choose a spot (e.g., the center of the whole area) as the venue to conduct sensing tasks. We also randomly choose a time window of five minutes, within which, there are $7,365$ vehicles reporting $54,395$ GPS coordinates. Considering distance is usually one of the main concerns for location-based sensing tasks (e.g., air pollution surveillance or noise monitoring) while drivers (or workers) carrying abundantly powered sensor devices (e.g., mobile phones) may more prefer to conduct crowdsensing tasks~\cite{GaoWXC-INFOCOM20,MullerTSK-TON18}, we choose \textit{distance to task spot} and \textit{battery state} as dimensions to construct a two-dimensional context space. Specifically, each worker could estimate its prospective trajectory in our time window and report the center of the trajectory to the crowdsensing platform. We assume the battery state for each worker obeys a uniform distribution in $[0,1]$. Given a worker with context $\phi=(\phi^{[0]}, \phi^{[1]})$ where $\phi^{[0]}$ and $\phi^{[1]}$ denote the distance to the task spot and the battery state, we define its sensing ability as $\mathbb E[r(\phi)] = \frac{1}{\sigma\sqrt{2\pi}} \exp \left( -\frac{({\phi^{[0]}})^2}{2\sigma^2} \right) \cdot \sqrt{\phi^{[1]}}$, by borrowing the idea from \cite{MullerTSK-TON18}. We let $\sigma=1$ such that the H$\ddot{\text{o}}$lder condition holds for $\alpha=2$ to facilitate our quantitative analysis. We then normalize the workers' sensing abilities into $[0,1]$. Note that our algorithm is compatible with an arbitrary mapping from context to sensing ability and the theoretical result shown in \textbf{Theorem}~\ref{thm:main} holds if the correlation between context and sensing ability respects the H$\ddot{\text{o}}$lder condition. Additionally, we adopt the same settings in terms of costs and capacities as our previous experiments on the synthetic dataset.

      Likewise, we first report in Fig.~\ref{fig:szbudget} the cumulative revenues and the regrets of the different algorithms with the budget varying from $2 \times 10^3$ to $3 \times 10^4$. It is illustrated that our algorithm yields higher cumulative revenue and smaller regret, compared with the other reference algorithms. We also evaluate the scalability of the different algorithms in handling an increasing number of unknown workers with fixed budget $B=1\times 10^4$. The results are reported in Fig.~\ref{fig:szworker}. Unsurprisingly, compared with other reference algorithms, our CAWS algorithm always has a much better performance in terms of both cumulative revenue and regret, no matter how many unknown workers are given. Furthermore, similar to what we have shown in Sec.~\ref{ssec:synthetic}, when the number of unknown workers is increased, our CAWS algorithm has almost the same performance such that only a slight increase can be observed in terms of cumulative revenue and regret respectively, whereas the reference algorithms obviously yield less cumulative revenue and have higher regret.
      \begin{figure}[htb!]
      \begin{center}
        \parbox{.49\textwidth}{\center\includegraphics[width=.48\textwidth]{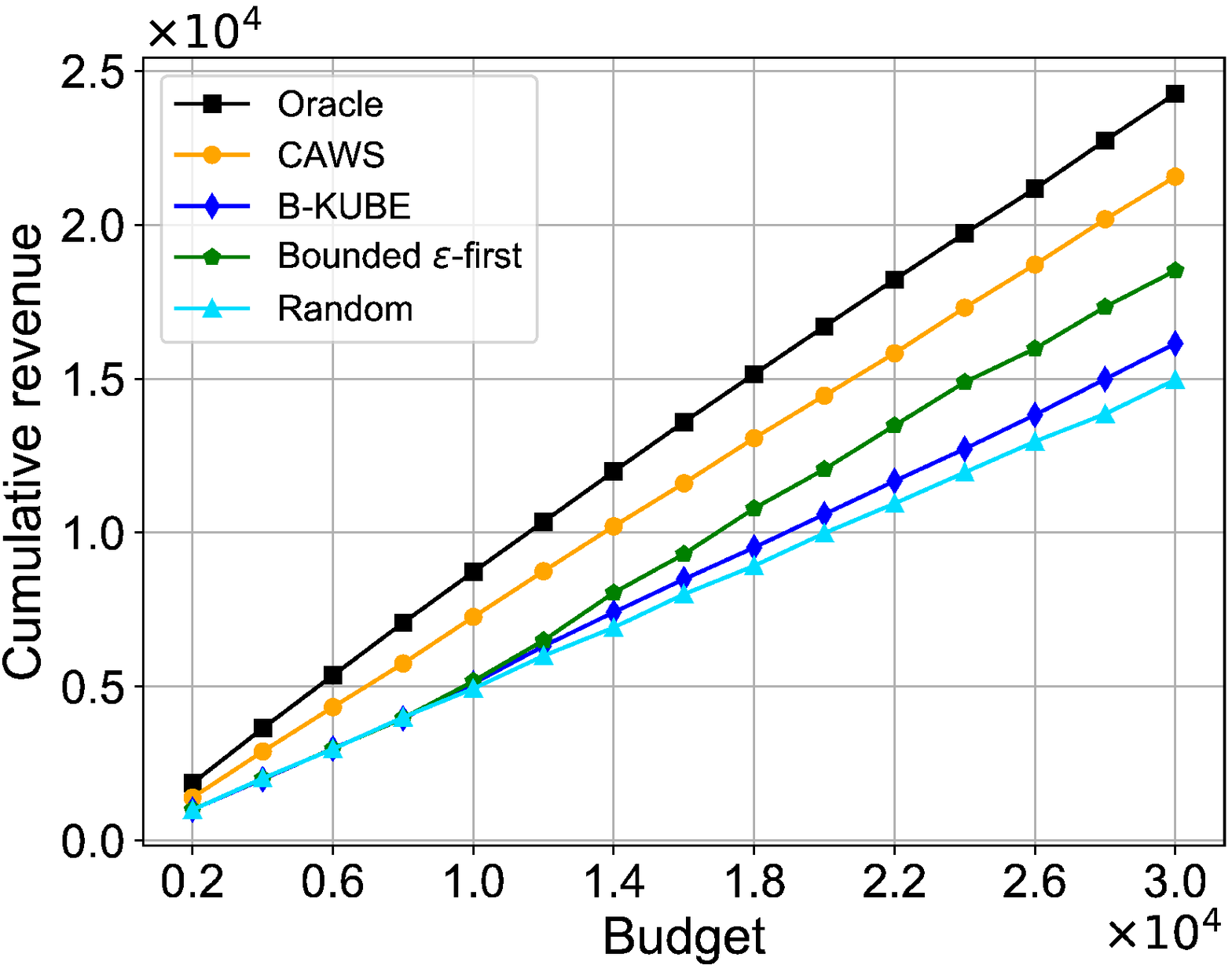}}
        \parbox{.49\textwidth}{\center\includegraphics[width=.48\textwidth]{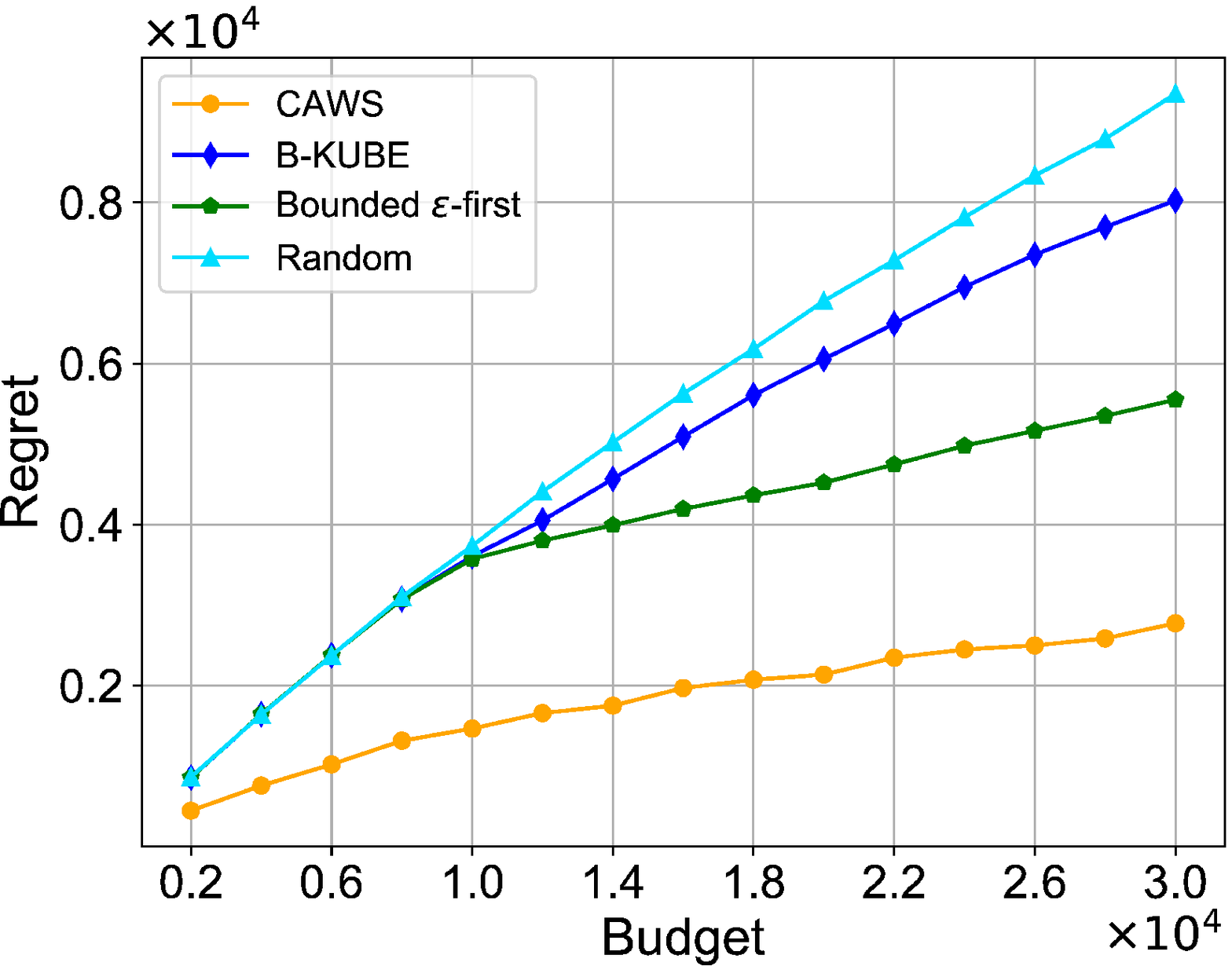}}
        \parbox{.49\textwidth}{\center\scriptsize(a) Cumulative revenue}
        \parbox{.49\textwidth}{\center\scriptsize(b) Regret}
        \caption{Comparisons with different budgets on vehicular trajectory dataset.}
      \label{fig:szbudget}
      \end{center}
      \end{figure}
      \begin{figure}[htb!]
      \begin{center}
        \parbox{.49\textwidth}{\center\includegraphics[width=.48\textwidth]{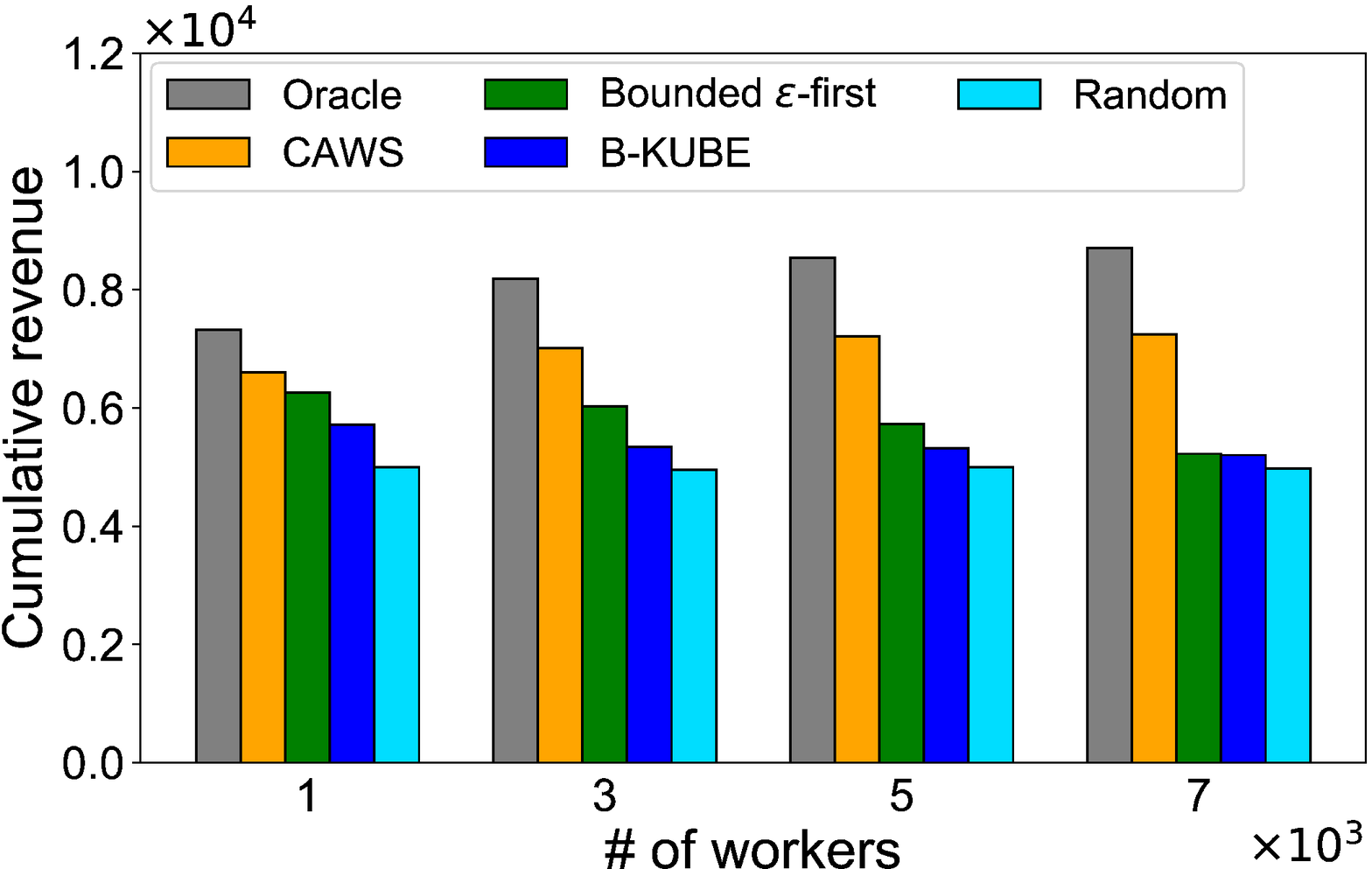}}
        \parbox{.49\textwidth}{\center\includegraphics[width=.48\textwidth]{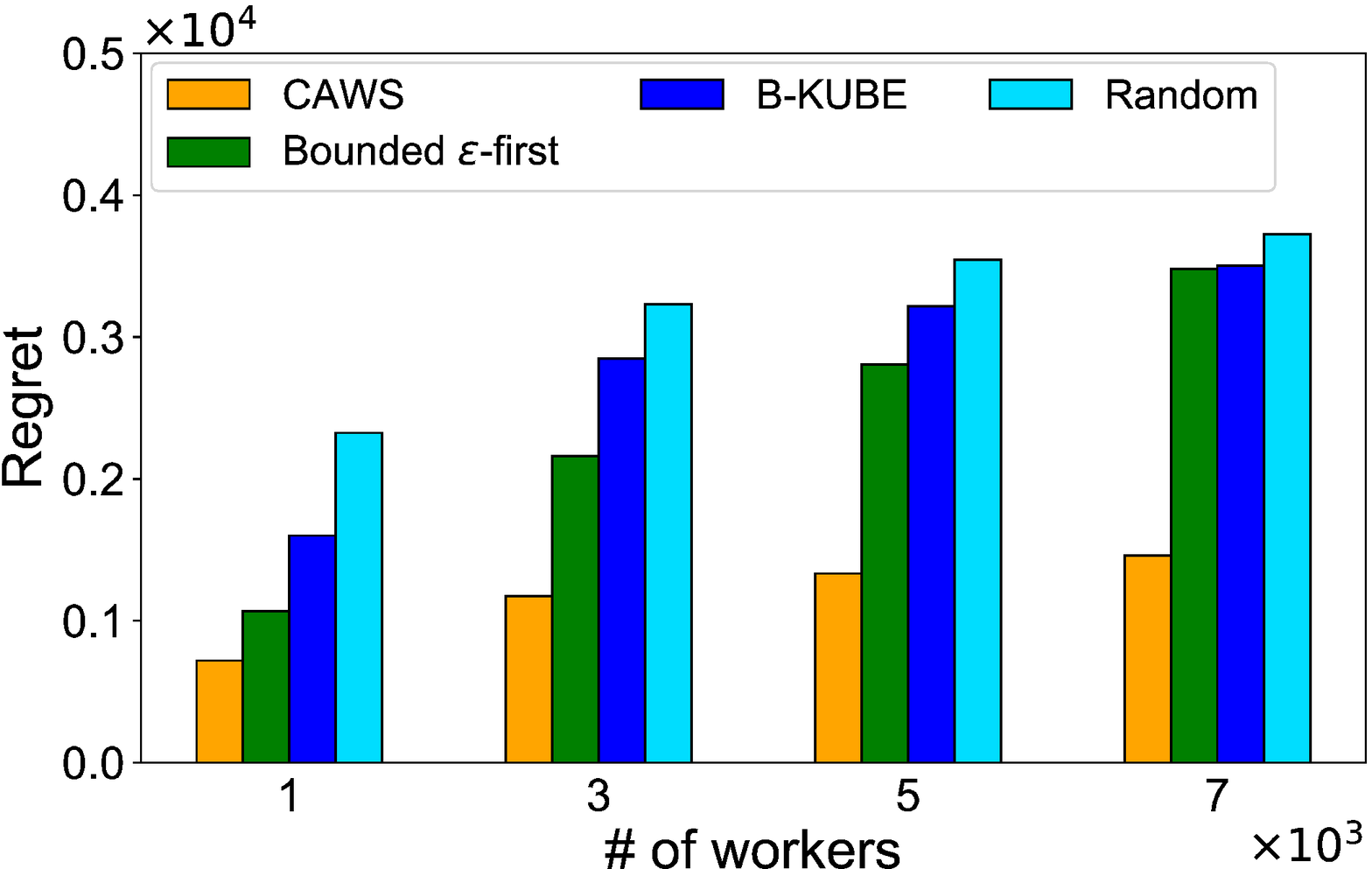}}
        \parbox{.49\textwidth}{\center\scriptsize(a) Cumulative revenue}
        \parbox{.49\textwidth}{\center\scriptsize(b) Regret}
        \caption{Comparisons with different numbers of workers on vehicular trajectory dataset.}
      \label{fig:szworker}
      \end{center}
      \end{figure}

    \subsubsection{Experiments on Yelp Dataset} \label{sssec:yep}
      In this section, we evaluate the performance of our CAWS algorithm in a crowdsensing application based on the dataset published by Yelp~\cite{yelp}. In fact, crowdsensing is a general paradigm for ubiquitous sensing, and the dataset includes abundant real-world traces for emulating spatial crowdsensing where Yelp workers are employed to review (or ``sense'') local business.
    
      We randomly choose $10^5$ workers from the dataset. For each worker, we set the number of his/her reviews as his/her capacities. Since there is no cost parameters for the workers in the dataset, we choose the cost parameters uniformly in the range $[1, 1.5]$ at random. In the Yelp dataset, the sensed data (i.e., the reviews of the workers on business) is voted by reviewers. For each of the sensed data, we assume that we get a unit of reward if it receives at least three positive votes
      %
      %
      In another word, we have $r_i=1$ if the data reported by worker $i$ receives at least three positive votes and thus is qualified; otherwise, $r_i=0$. We choose \textit{number of fans}, \textit{number of friends} and \textit{number of years as elite} as the context dimensions, considering their strong correlations to the data quality. Due to the space limit, we take a two-dimensional context space as an example where we adopt \textit{number of fans} and \textit{number of friends} as the dimensions, and illustrate the data quality distribution in the context space in Fig.~\ref{fig:2features}. It is apparently observed that the two context dimensions are closely related to the data quality. In our experiments, we gradually increase the dimensionality of the context space, to evaluate our algorithm.
      \begin{figure}[htb] 
        \centering
          \includegraphics[width=.5\columnwidth]{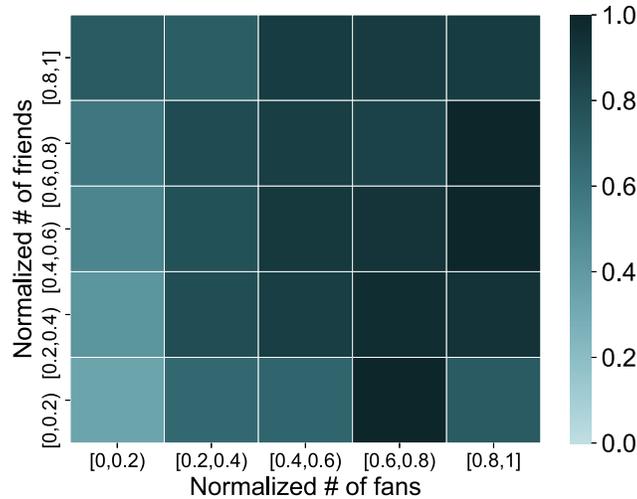}
        \caption{Data quality distribution in context space.}
        \label{fig:2features}
      \end{figure}
      Different from the synthetic dataset where $\alpha$ is controllable, we have to figure out an appropriate value for $\alpha$ when using the Yelp dataset to construct the context space, since $\alpha$ is an intrinsic parameter for real data. To quantitatively evaluate our algorithm, we first illustrate in Fig.~\ref{fig:alpha} the impact of different values of $\alpha$ on the performance of our algorithm. According to the results shown in Fig.~\ref{fig:alpha}, we set $\alpha = 2, 0.75, 0.25$ for $M=1,2,3$ respectively in the following. It is worthy to note that our algorithm still work with arbitrary $\alpha$ and we hereby seek for an appropriate value for $\alpha$ only for the purpose of quantitative evaluation.
      \begin{figure*}[htb!] 
        \begin{center}
          \parbox{.32\textwidth}{\center\includegraphics[width=.31\textwidth]{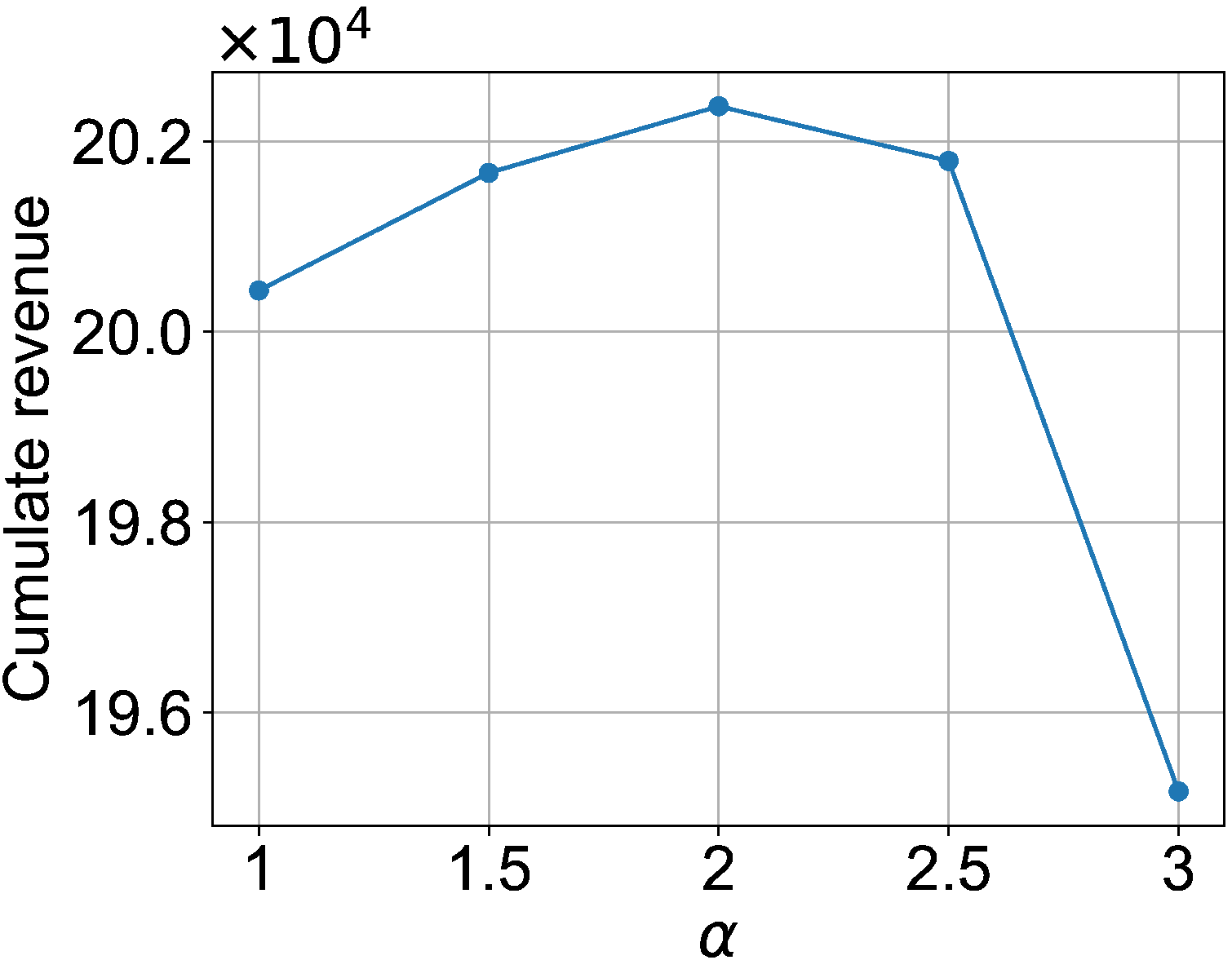}}
          \parbox{.32\textwidth}{\center\includegraphics[width=.31\textwidth]{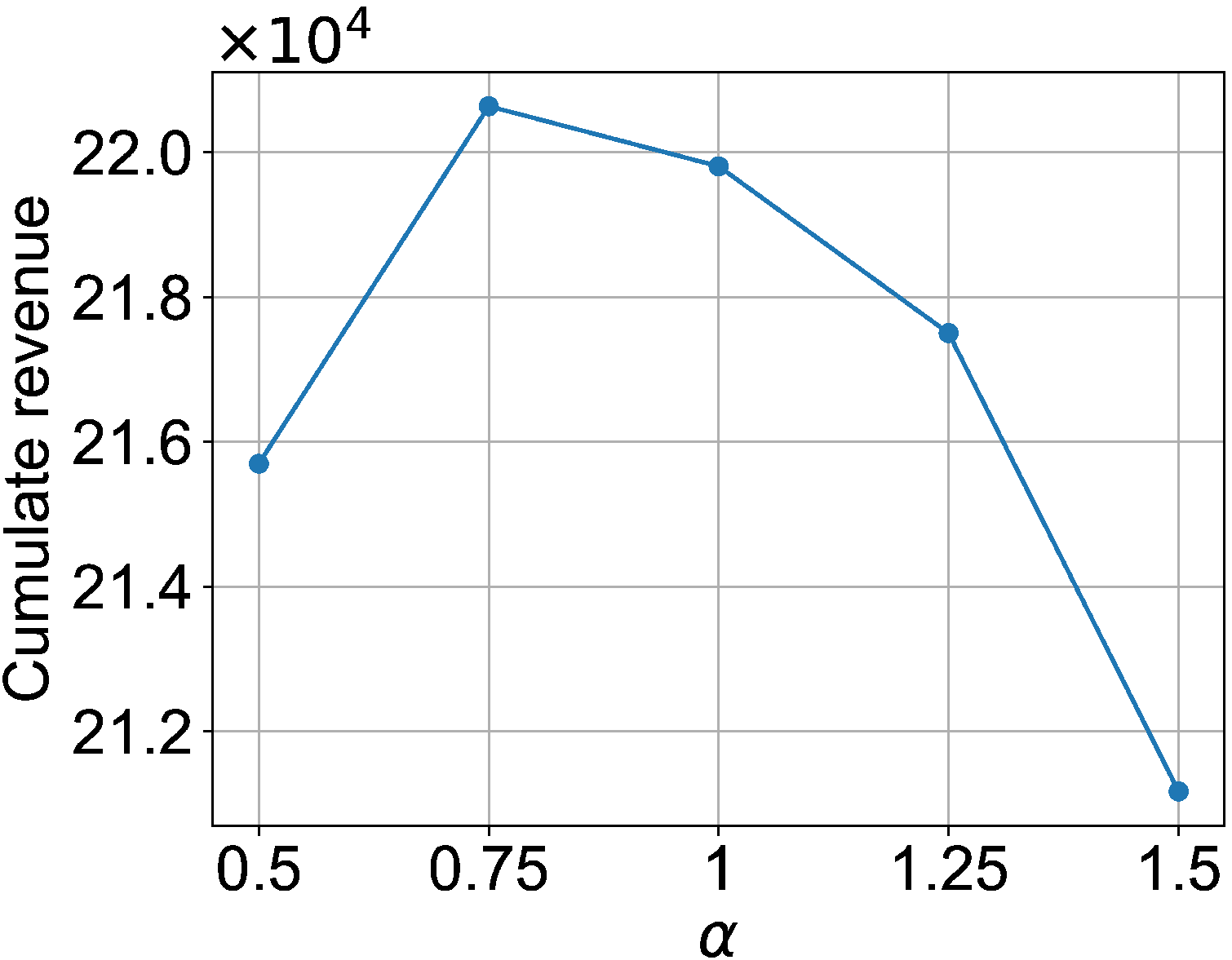}}
          \parbox{.32\textwidth}{\center\includegraphics[width=.31\textwidth]{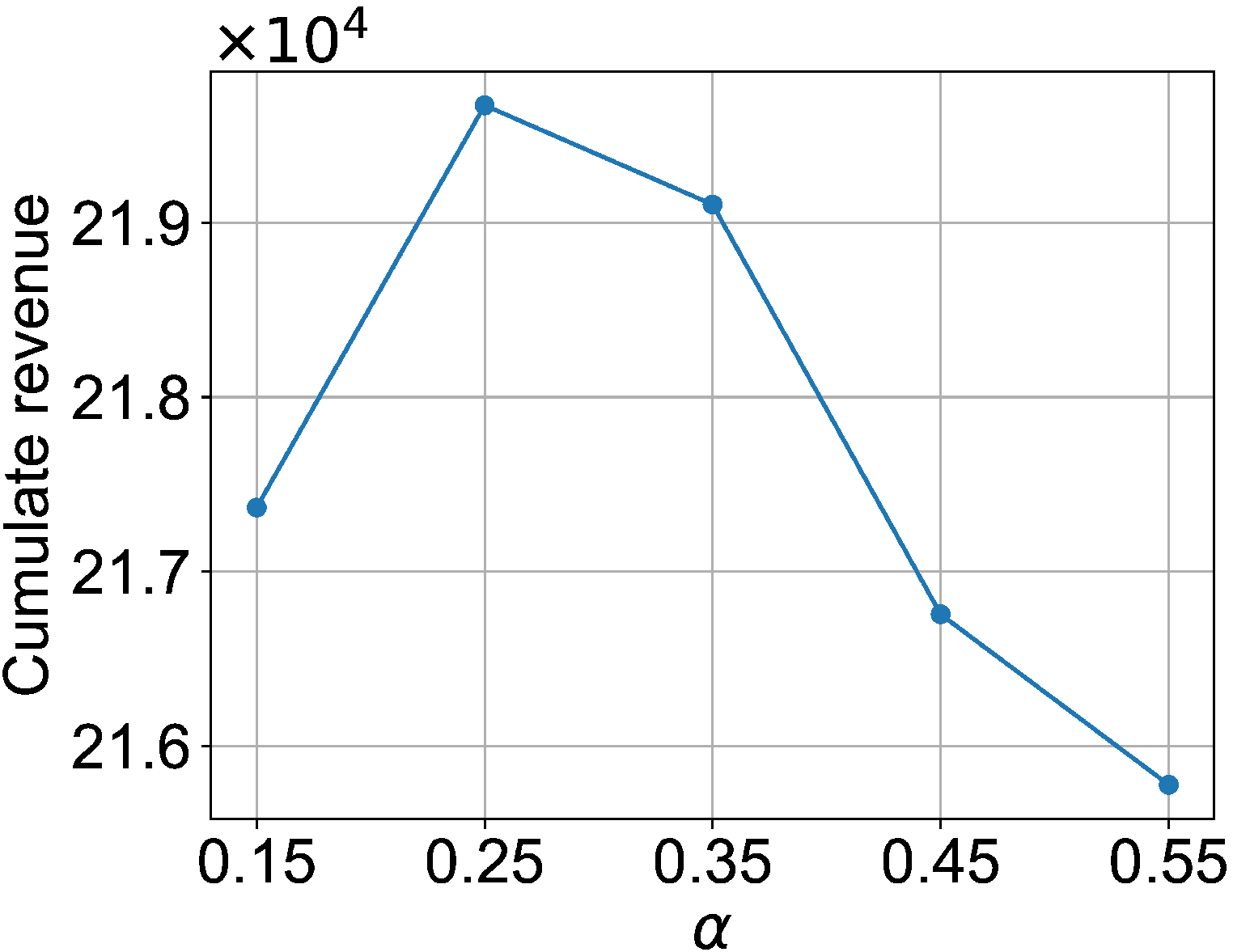}}
          \parbox{.32\textwidth}{\center\scriptsize(a) $M=1$}
          \parbox{.32\textwidth}{\center\scriptsize(b) $M=2$}
          \parbox{.32\textwidth}{\center\scriptsize(c) $M=3$}
          \caption{Cumulative revenues under different values of $\alpha$.}
          \label{fig:alpha}
        \end{center}
      \end{figure*}
    
      Since the dataset does not include the (expected) sensing abilities of the workers, we focus on investigating the performances of the algorithms in terms of cumulative revenue only. For each worker, when it is selected, we randomly choose one from its data samples without replacement to calculate the cumulative revenue. In addition, we vary budget $B$ from $2 \times 10^4$ to $4 \times 10^5$ with a step size $2 \times 10^4$ to show the performance of the algorithms under different budgets. It is shown by the results in Fig.~\ref{fig:yelpcrevenue} that, our CAWS algorithm outperforms the others and its performance is very close to the ones of the oracle (for all $M=1,2,3$), especially under less budget. Furthermore, since our CAWS algorithm adaptively tunes the granularity to partition the context space mainly according to the number of dimensions and the budget, it results in similar cumulative revenues in all the three context spaces. By taking into account more relevant dimensions (e.g., by increasing $M$ from $1$ to $2$), our algorithm yields more cumulative revenue. Nevertheless, a higher-dimensional context space does not always imply much higher cumulative revenue. For example, the resulting cumulative revenue in the three-dimensional context space is very close to the one in the two-dimensional context space.
      \begin{figure}[htb!]
        \centering
        \includegraphics[width=.5\columnwidth]{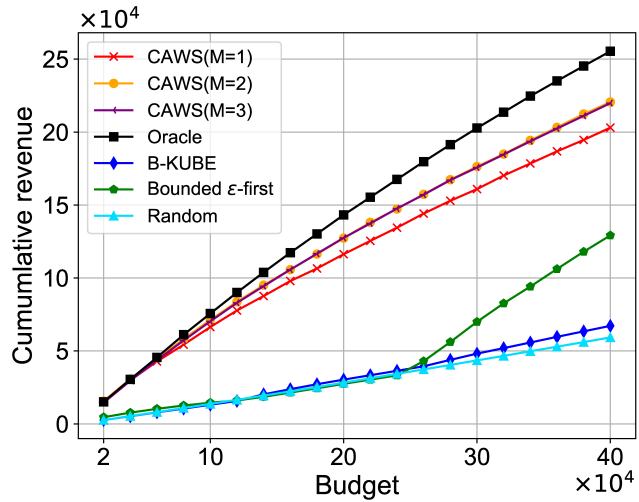}
        \caption{Comparisons with different budgets on Yelp dataset.}
        \label{fig:yelpcrevenue}
      \end{figure}

      We also evaluate the algorithms in terms of scalability to the different numbers of workers. We fix budget $B=1 \times 10^5$ while varying the number of workers from $4 \times 10^4$ to $1\times 10^5$ with a step size $2 \times 10^4$. As illustrated by the results in Fig.~\ref{fig:yelpdiffworker}, our CAWS algorithm yields much higher cumulative revenue than the others in all settings. Furthermore, as the number of workers is increased, the performance of our algorithm is always close to the one of oracle, while the others yield decreased cumulative revenues. Especially, when the number of worker is huge, e.g., $N=10^5$, our algorithm can yield six times higher sensing revenue than the other ones. Similar to our observations in Fig.~\ref{fig:yelpcrevenue}, the cumulative revenue obtained by applying our algorithm in the two-dimensional context space is very close to the one yielded by our algorithm in the three-dimensional context space. Additionally, since we partition the context space with a carefully tuned granularity, the performance of our algorithm always can be ensured when we introduce much more workers with limited budget. 
      \begin{figure}[htb!]
        \centering
        \includegraphics[width=.5\columnwidth]{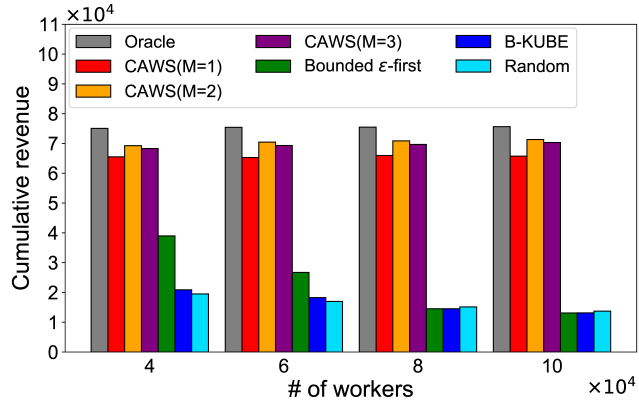}
        \caption{Comparisons with different numbers of workers on Yelp dataset.}
        \label{fig:yelpdiffworker}
      \end{figure}

  \subsection{An Extension to Time-Varying Context} \label{ssec:ext}
    We use the vehicular trajectory dataset again to evaluate the performance of our CAWS algorithm in assigning a sensing task to a large number of unknown workers with time-varying contexts. We let the center of the whole area be the spot of interest and consider a time span consisting of $1,300$ time windows. We randomly choose $500$ vehicles as workers, each of which reports at least $5$ GPS coordinates regularly. We additionally assume that, for each worker, the battery state of its mobile phone is decreased by $5\%$ in every time window and the mobile phone is recharged when its power is exhausted. The sensing ability for each worker is defined as the same as shown in Sec.~\ref{sssec:szstatic}. We adopt the same settings in terms of costs and capacities as before. 

    Likewise, we adopt cumulative revenue and regret as the metrics for the purpose of evaluation. It is worthy to note that, when calculating the regret, the (off-line) rounding-based density-ordered greedy algorithm takes the vehicles' trajectories (and thus their time-varying sensing abilities) across the time span as input according to the definition of the regret. We first evaluate the different algorithms with the budget varying from $300$ to $1500$. Since we investigate the correlation between context and sensing ability in our CAWS algorithm such that the workers' sensing abilities can be estimated according to their instant contexts, while the others have to re-learn a worker's sensing ability once its context (and thus its sensing ability) is changed, our CAWS algorithm has obvious advantages over the others in terms of both cumulative revenue and regret, as shown in Fig.~\ref{fig:szbudgetdy}. We also demonstrate the performance of the algorithms with different numbers of workers in Fig.~\ref{fig:szworkertdy}. The number of workers is varied from $200$ to $500$, while the budget is fixed to $1,500$. Similar to the experiment results where the workers have static contexts (see Sec.~\ref{ssec:synthetic} and Sec.~\ref{ssec:real}), increasing the number of workers cannot let our algorithm loose its advantages over the other three alternative algorithms. Moreover, when the number of workers is increased, the performance of our CAWS algorithm is stable under the given budget, while the ones of the other three opponents are degraded, especially in terms of regret.
    \begin{figure}[htb!]
    \begin{center}
      \parbox{.48\textwidth}{\center\includegraphics[width=.48\textwidth]{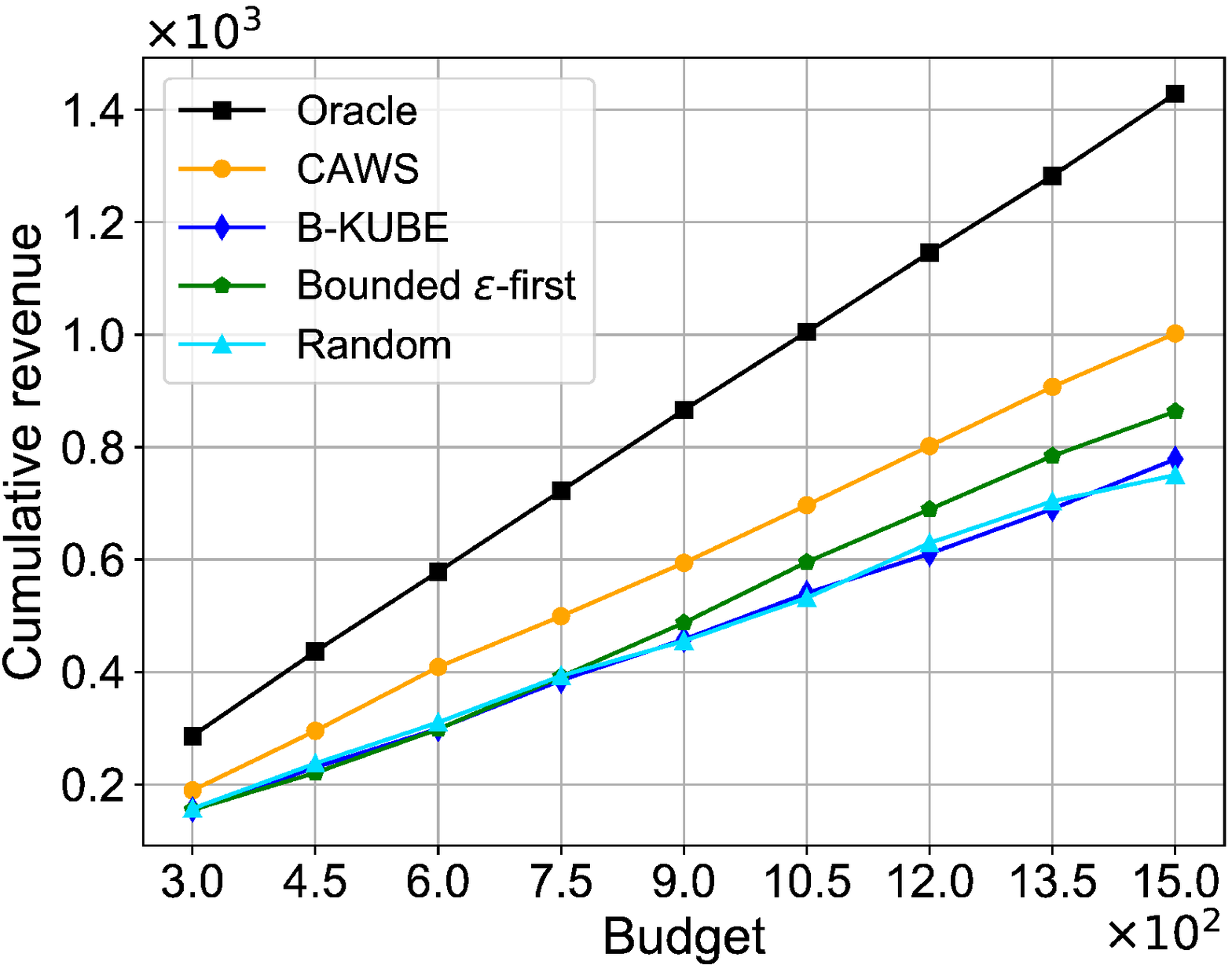}}
      \parbox{.48\textwidth}{\center\includegraphics[width=.48\textwidth]{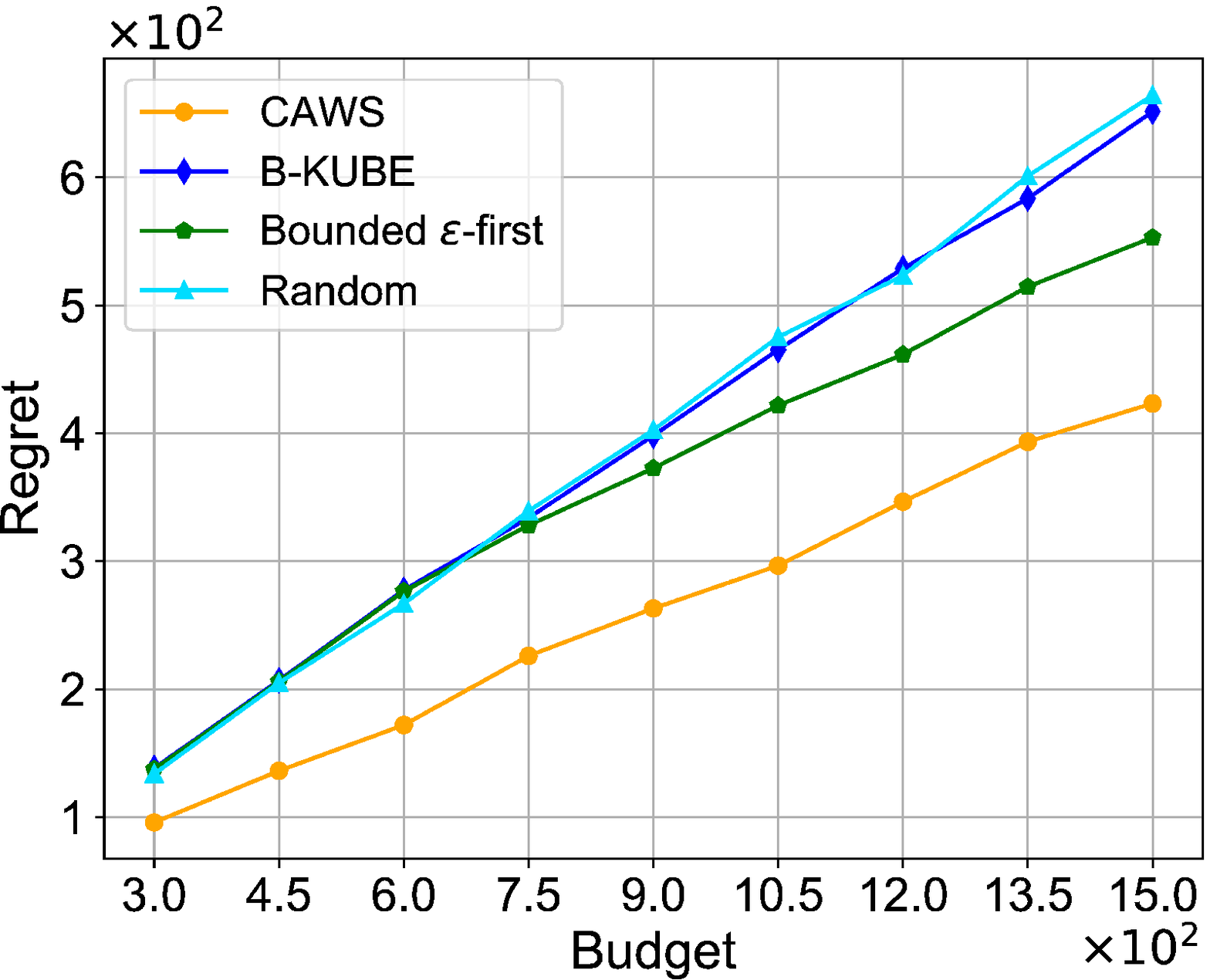}}
      \parbox{.48\textwidth}{\center\scriptsize(a) Cumulative revenue}
      \parbox{.48\textwidth}{\center\scriptsize(b) Regret}
    \end{center}
    \caption{Comparisons under different budgets with time-varying context.}
    \label{fig:szbudgetdy}
    \end{figure}
    \begin{figure}[htb!]
    \begin{center}
      \parbox{.48\columnwidth}{\center\includegraphics[width=.48\columnwidth]{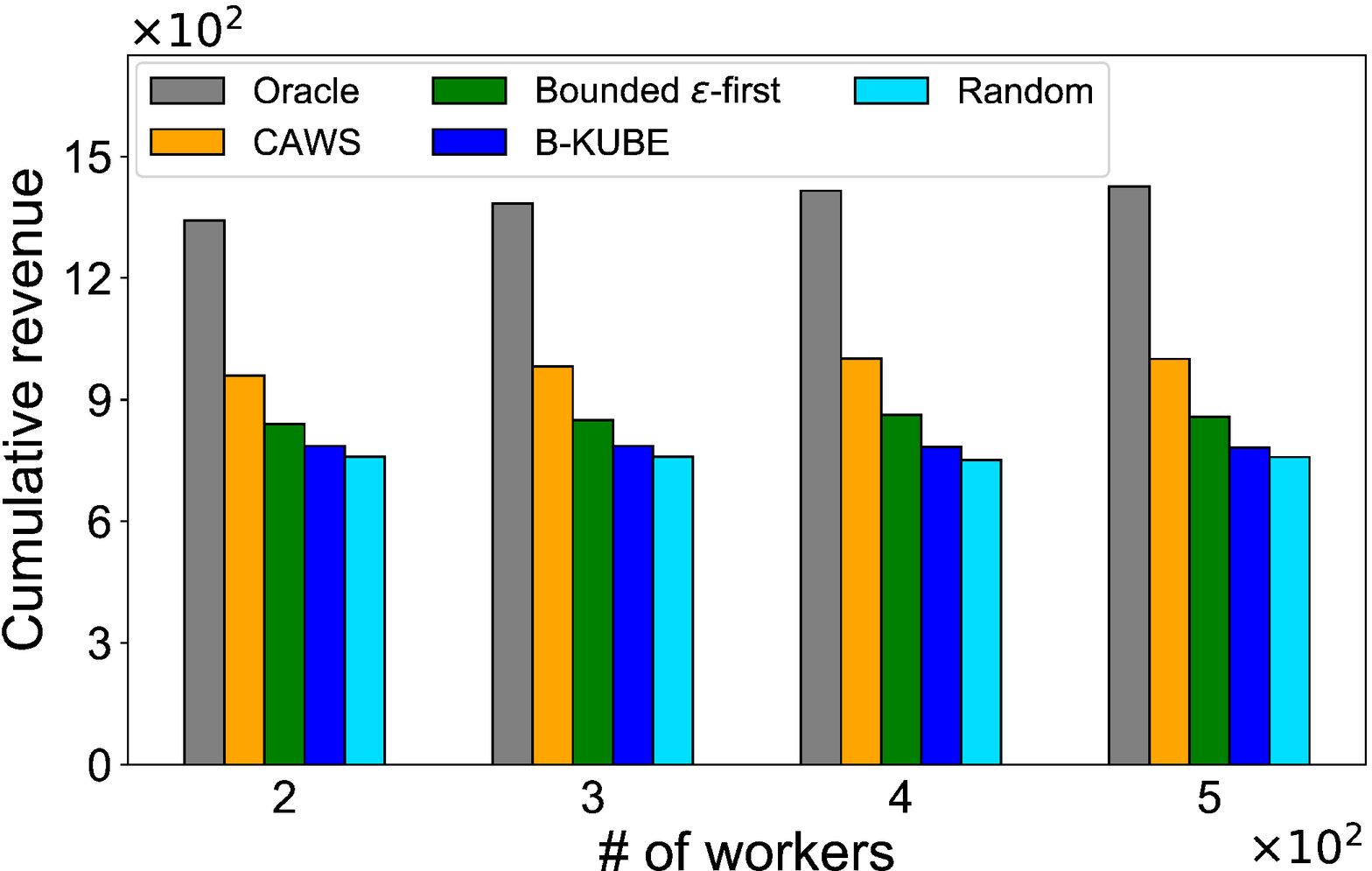}}
      \parbox{.48\columnwidth}{\center\includegraphics[width=.48\columnwidth]{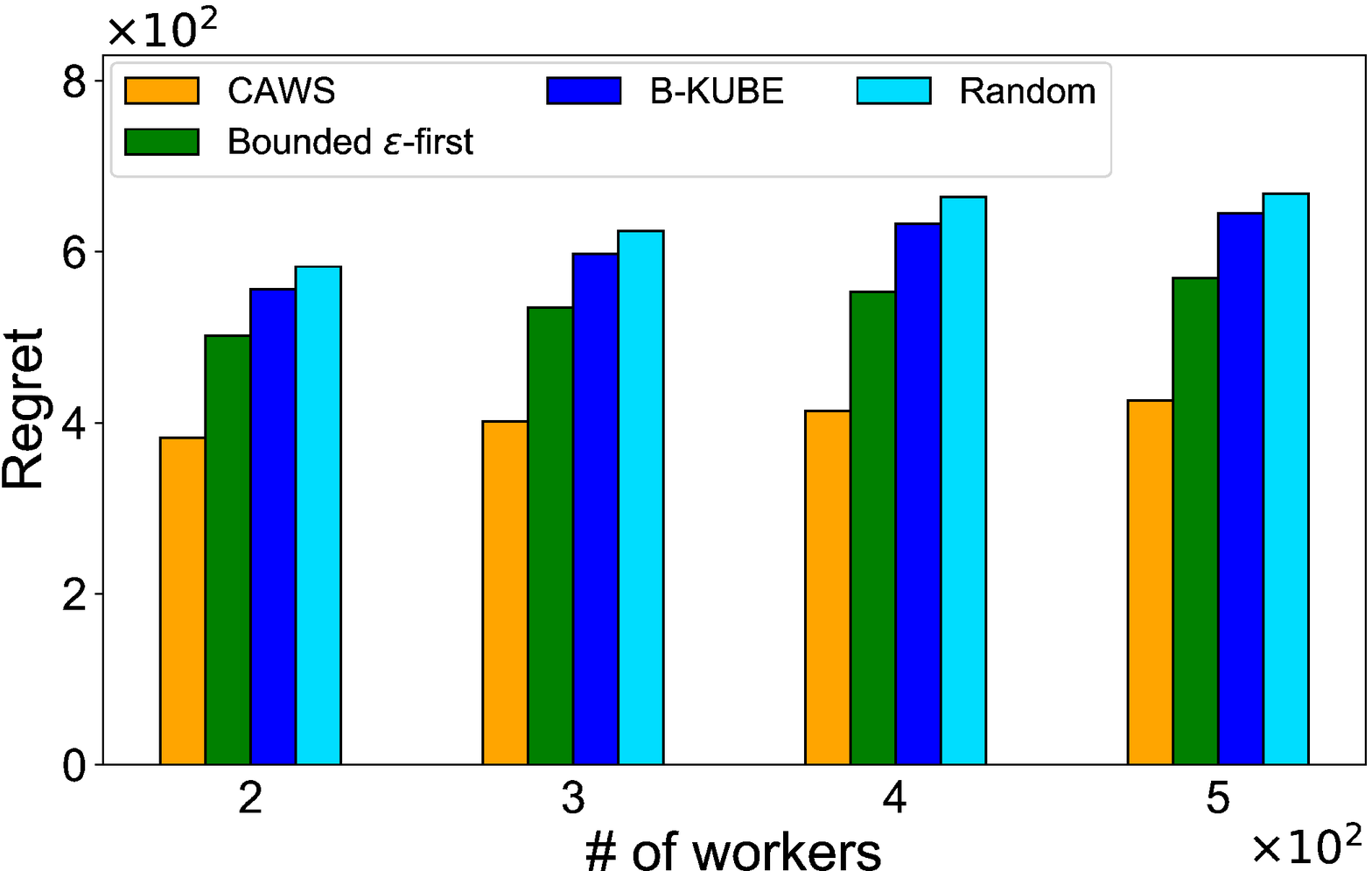}}
      \parbox{.48\columnwidth}{\center\scriptsize(a) Cumulative revenue}
      \parbox{.48\columnwidth}{\center\scriptsize(b) Regret}
    \end{center}
    \caption{Comparisons under different numbers of workers with time-varying context.}
    \label{fig:szworkertdy}
    \end{figure}

\section{Related Work} \label{sec:relwork}
  In the past decades, there have been a vast body of studies on the fundamental problem of worker selection in crowdsensing systems \cite{SongLWMW-TVT14,LiLW-MASS15,PuCXF-INFOCOM16,HanHL-TON18,YangLWW-TMC19}.
  %
  %
  %
  However, most of the existing proposals assume that the workers' sensing abilities are known as prior, while such an assumption may not be the case in practice. Therefore, there have been a few recent studies considering the uncertain worker selection problems where the worker's sensing abilities are unknown. For example, \cite{HanZL-TON16,RangiF-AAMAS18} study the uncertain selection problem such that the workers with unknown sensing abilities are selected sequentially under a limited total budget to perform a given sensing task. In \cite{ZhaoXWXHZ-TMC21}, unknown workers are selected sequentially with their sensing abilities being private information to be preserved. \cite{SongJ-INFOCOM21} adopts the empirical entropy of the data reported by workers to measure the sensing revenue. In \cite{GaoWYXC-ICPADS19}, a multi-task assignment problem is investigated. Therein, unknown workers are selected to maximize the sensing revenue, such that the resulting total cost does not exceed the budget and all the sensing tasks can be completed. The multi-task assignment problem is also studied in \cite{GaoWXC-INFOCOM20}. Each worker first submits its options (i.e., a subset of the tasks), and the crowdsensing platform assigns one of the options to each worker under a given budget, aiming at maximizing the sensing quality. As mentioned in Sec.~\ref{sec:intro}, the arms in the CMAB framework (e.g., the workers in \cite{HanZL-TON16,RangiF-AAMAS18,ZhaoXWXHZ-TMC21,SongJ-INFOCOM21} or the worker-task combinations in \cite{GaoWYXC-ICPADS19,GaoWXC-INFOCOM20}) are exploited and explored individually. Therefore, these conventional CMAB-based algorithms are of low efficiency especially when the number of arms is huge while the budget is limited, as shown in Sec. \ref{ssec:dis} and \ref{sec:exp}.

  Context information is very useful for crowdsensing systems and has been extensively utilized in designing worker selection algorithms \cite{YururLSLML-COMST16,NejadAMM-FGCS19}. In \cite{HassaniHJ-ICPADS15}, a similarity model is designed to calculate the context similarities between tasks and workers, and a worker is said to have higher sensing ability if its context is more similar to the one of the sensing task. The sensing tasks are then assigned to the workers according to their eligibilities so as to improve sensing efficiency. The matching between sensing tasks and workers is also studied in \cite{YucelYB-TMC2021} where the requirements of tasks and the preferences of workers are considered. Recently, \cite{HanYYWYB-TMC21} proposes novel data structures to improve the performance of the task-worker matching. Although the contexts of the workers are utilized to guide the assignment of the sensing tasks, the uncertainties of the workers are not taken into account in these proposals. In \cite{LiuZWTC-INFOCOM17}, a context-based data quality classifier is trained from historical data in an off-line manner, according to which, the workers are selected. Although machine learning methods are applied to train the classifier, it entails a large set of off-line training samples as input and takes into account neither budget constraints nor capacity constraints. In \cite{MullerTSK-TON18}, the dependence of workers' sensing abilities on both the workers' and the given tasks' context information is learned in an on-line manner. Although a budget constraint is considered in \cite{MullerTSK-TON18}, it assumes the task requester has a fixed budget for each selection decision, while our algorithm considers a strict total budget constraint for the whole learning process. Moreover, \cite{MullerTSK-TON18} does not consider the bounded capacities of the workers.
  %

  MAB problem is a typical reinforcement learning problem and has been studied for several decades. So far, several well-known algorithms, e.g., $\epsilon$-greedy algorithm, UCB algorithm, etc, have been proposed~\cite{LaiR-AAM85,AuerCF-ML02}. It is then extended to CMAB problem, to address the uncertainties in combinatorial optimization problems \cite{GaiKJ-TON12,TranCRJ-AAAI12,ChenWY-ICML13,ThanhSRJ-AI14,RangiF-AAMAS18}. Motivated by contextual bandit where the context information of the arms is utilized \cite{LangfordZ-NIPS08,LiCLS-WWW10}, a contextual CMAB framework is proposed in \cite{ChenXL-NIPS18}, which inherits from both contextual bandit and combinatorial bandit. Specifically, it studies the budget-limited worker selection problem within a given time horizon. In each time slot, it allocates a fixed amount of budget to either exploit or explore a group of workers. Therefore, it cannot be applied to our problem where the total budget is limited such that we have to make full use of the budget to discriminate the workers with uncertain sensing abilities. Also, it cannot handle the capacity constraints of the workers, whereas these constraints are the main concerns of our CAWS algorithms.

\section{Conclusion and Future Work} \label{sec:conclusion}
  In this paper, we have studied how to select among a massive number of uncertain workers with bounded sensing capacities under a limited budget, such that the expected cumulative sensing revenue can be maximized with both the budget constraint and the capacity constraints respected. Although the conventional CMAB framework can be applied to address the above problem by exploring and exploiting the workers individually, it is of quite low efficiency when the number of workers is rather huge while the budget is significantly limited. To address the above issue, we have proposed a worker selection algorithm, i.e., CAWS, which makes a trade-off between exploitation and exploration in a context space instead of among the individual workers. We have performed a rigorous theoretical analysis to prove the regret of our CAWS algorithm is upper-bounded by $\mathcal O(B^{\frac{M}{\alpha+M}} \ln B)$ through partitioning the context space with a fine-tuned granularity. We also have conducted extensive experiments using both synthetic and real datasets to verify the considerable advantages of our CAWS algorithm over the existing state-of-the-art algorithms.

  As we have shown in Sec.~\ref{ssec:dis} and Sec.~\ref{ssec:ext}, harnessing context is a very promising method to select among massive workers with unknown and time-varying sensing abilities, given a significantly limited budget. Nevertheless some preliminary results have been presented in this paper, we are on the way of building a rigorous theoretical framework to design and analyze competitive algorithms for the budget-limited worker selection problem with online uncertainties.

\bibliographystyle{IEEEtran}
\bibliography{cknapsack} 

\begin{appendices}
\section{Proof of Lemma~\ref{le:iterations}}
  \begin{lemma} \label{le:ineq0}
    Suppose our CAWS algorithm proceeds $T$ iterations and let $B(t)$ denotes the residual budget at the beginning of the $t$-th iteration. Initially, $B(1)=B$. For each iteration $t=1,2,\cdots,T$, we have
    \begin{equation} \label{eq:ineq0}
      \frac{c_{min}}{B(t)} \leq \frac{1}{T-t+1} 
    \end{equation}
  \end{lemma}
  \begin{proof}
    At the beginning of the $t$-th iteration, the residual budget is $B(t)$. Since we select the workers $T$ times in total, for any $1 \leq t \leq T$,
    \begin{equation*}
      (T-t+1) c_{min} \leq c_{i(t)} + c_{i(t+1)} + \cdots + c_{i(T)} \leq B(t)
    \end{equation*}
    based on which, we have the inequality (\ref{eq:ineq0}).
  \end{proof}

  Assume $\lfloor \tilde{\mathbf{x}}^*(t) \rfloor$ and $\lfloor \hat{\mathbf{x}}^*(t) \rfloor$ denote the solutions by applying the rounding-based density-ordered greedy algorithm to the BKP instances $(i, \mu_{Q_i}, c_i, \tau_i(t), B(t))^N_{i=1}$ and $(i, U_i(t), c_i, \tau_i(t), B(t))^N_{i=1}$, respectively. By replacing $\mu_{Q_i}$ with $U_i(t)$, $\lfloor \hat{\mathbf{x}}^*(t) \rfloor$ can be considered as an estimate on $\lfloor \tilde{\mathbf{x}}^*(t) \rfloor$.
  \begin{lemma} \label{le:swap}
    If there is a worker $j$ such that $j \notin \lfloor \tilde{\mathbf{x}}^* \rfloor$ and $j \in \lfloor \hat{\mathbf{x}}^*(t) \rfloor$, then there is at least one worker $j' \in \lfloor \tilde{\mathbf{x}}^* \rfloor$ such that
    \begin{equation} \label{eq:swapineq1}
      \frac{\mu_{Q_j}}{c_j} \leq \frac{\mu_{Q_{j'}}}{c_{j'}}
    \end{equation}
    and
    \begin{equation} \label{eq:swapineq2}
      \frac{1}{c_j}\left( \bar r_{Q_j}(t) + \sqrt{\frac{2\log t}{\lambda_{Q_j}(t)}} \right) \geq \frac{1}{c_{j'}}\left( \bar r_{Q_{j'}}(t) + \sqrt{\frac{2\log t}{\lambda_{Q_{j'}}(t)}} \right)
    \end{equation}
    where $Q_j \neq Q_{j'}$
    Also, the worker $j'$ has non-zero residual capacity to perform sensing tasks.
  \end{lemma}
  \begin{proof}
    If a worker $j \notin \lfloor \tilde{\mathbf{x}}^* \rfloor$, then the worker $j \notin \lfloor \tilde{\mathbf{x}}^*(t) \rfloor$, since $B(t) \leq B$. Moreover, according to the procedures of the rounding-based density-ordered greedy algorithm, if the worker $j \in \lfloor \hat{\mathbf{x}}^*(t) \rfloor$, there exists at least one work $j' \in \lfloor \tilde{\mathbf{x}}^*(t) \rfloor$ such that $\frac{\mu_{Q_j}}{c_j} \leq \frac{\mu_{Q_{j'}}}{c_{j'}}$ and $\frac{U_j(t)}{c_j} \geq \frac{U_{j'}(t)}{c_{j'}}$. Also, $Q_{j} \neq Q_{j'}$; otherwise, we would have both $c_j \geq c_{j'}$ and $c_j \leq c_{j'}$ hold since $\mu_{Q_j} = \mu_{Q_{j'}}$ and $U_j(t) = U_{j'}(t)$. Also, $j' \in \lfloor \tilde{\mathbf{x}}^*(t) \rfloor$ implies that the worker $j'$ has non-zero residual capacity to perform additional tasks. Furthermore, since $\lfloor \tilde{\mathbf{x}}^*(t) \rfloor \subseteq \lfloor \tilde{\mathbf x}^* \rfloor$, $j' \in \lfloor \tilde{\mathbf x}^* \rfloor$ and $j'$ can perform more tasks.
  \end{proof}

  \begin{lemma} \label{le:upbdj}
    Assume our CAWS algorithm proceeds $T$ iterations. For $\forall Q \in \Omega$, we have
    \begin{align} \label{eq:upbdj}
      \mathbb P( i(t) \in \mathcal{N}^{-}_Q \mid T) \leq& \mathbb P(i(t) \in \lfloor \hat{\mathbf{x}}^*(t) \rfloor, i(t) \in \mathcal{N}^{-}_Q \mid T ) +  \left({c_{max}}/{c_{min}}\right)^2/{(T-t+1)}
    \end{align}
  \end{lemma}
  \begin{proof}
    Since
    \begin{align} \label{eq:decom}
      \mathbb P( i(t) \in \mathcal{N}^{-}_Q \mid T) &= \mathbb P(i(t) \in \mathcal{N}^{-}_Q, i(t) \in \lfloor \hat{\mathbf{x}}^*(t) \rfloor \mid T) + \mathbb P( i(t) \in \mathcal{N}^{-}_Q, i(t) \notin \lfloor \hat{\mathbf{x}}^*(t) \rfloor \mid T)
    \end{align}
    we can prove this lemma by deriving the upper-bound of $\mathbb P( i(t) \in \mathcal{N}^{-}_Q, i(t) \notin \lfloor \hat{\mathbf{x}}^*(t) \rfloor \mid T)$

    Recall that $\lfloor \hat{\mathbf{x}}^*(t) \rfloor$ denotes the solution of the BKP instance $(i, U_i(t), c_i, \tau_i(t), B(t))^N_{i=1}$ by the rounding-based density-ordered greedy algorithm. Let $\hat k(t)$ denote the split worker. Then, after selecting the worker $\hat k$, the residual budget is less than or equal to $c_{\hat k(t)}$; therefore,
    \begin{equation}
      \sum_{i \notin \lfloor \hat{\mathbf{x}}^*(t) \rfloor} x_i(t) \leq \frac{c_{\hat k(t)}}{c_{min}} \leq \frac{c_{max}}{c_{min}}
    \end{equation}
    Furthermore, considering the selection outputted by our density-ordered greedy subroutine can be bounded as $\sum_{i\in \mathcal N} x_i(t) \geq \frac{B(t)}{c_{max}}$, we have
    \begin{equation}
      \frac{\sum_{i \notin \lfloor \hat{\mathbf{x}}^*(t) \rfloor} x_i(t)}{\sum_{i\in \mathcal N} x_i(t)} \leq \frac{c_{\hat k(t)}}{c_{min}} \leq \frac{c_{max}}{c_{min}} \cdot \frac{c_{max}}{B(t)}
    \end{equation}
    By substituting the inequality (\ref{eq:ineq0}) in \textbf{Lemma}~\ref{le:ineq0} into the above inequality, we have
    \begin{equation}
      \frac{\sum_{i \notin \lfloor \hat{\mathbf{x}}^*(t) \rfloor} x_i(t)}{\sum_{i\in \mathcal N} x_i(t)} \leq \left( \frac{c_{max}}{c_{min}} \right)^2 \cdot \frac{1}{T-t+1}
    \end{equation}
    Then, the upper-bound of $\mathbb P( i(t) \in \mathcal{N}^{-}_Q, i(t) \notin \lfloor \hat{\mathbf{x}}^*(t) \rfloor \mid T)$ can be derived as follows
    \begin{align}
      & \mathbb P( i(t) \in \mathcal{N}^{-}_Q, i(t) \notin \lfloor \hat{\mathbf{x}}^*(t) \rfloor \mid T) \nonumber\\
      \leq& \mathbb P(i(t) \notin \lfloor \hat{\mathbf{x}}^*(t) \rfloor \mid T) \nonumber\\
      =& \sum_{\{x_i(t)\}^N_{i=1}} \mathbb P(i(t) \notin \lfloor \hat{\mathbf{x}}^*(t) \rfloor \mid  \{x_i(t)\}^N_{i=1}, T) \cdot \mathbb P(\{x_i(t)\}^N_{i=1}) \nonumber\\
      \leq& \sum_{\{x_i(t)\}^N_{i=1}}  \left( \frac{c_{max}}{c_{min}} \right)^2 \cdot \frac{1}{T-t+1} \cdot \mathbb P(\{x_i(t)\}^N_{i=1}) \nonumber\\
      =& \left({c_{max}}/{c_{min}}\right)^2/{(T-t+1)} 
    \end{align}
    By substituting which into (\ref{eq:decom}), we complete the proof.
  \end{proof}

  \begin{lemma} \label{le:upbdy}
    For $\forall Q\in \Omega$, let $Y_Q(t) = \sum^t_{t'=1} \mathbb I(i(t')\in \mathcal N^-_Q)$ denote the number of times the workers in $\mathcal N^-_Q$ is selected by our CAWS algorithm up to the $t$-th iteration. We then have
    %
    \begin{eqnarray} \label{eq:upbdy1}
      \mathbb E_{\{i(t)\}^T_{t=1}}[Y_Q(T) \mid T]  \leq \xi \ln T + \frac{\pi^2}{3} + 1
    \end{eqnarray}
  \end{lemma}
  %
  \begin{proof}
    According to \textbf{Lemma}~\ref{le:upbdj}, $\mathbb E_{\{i(t)\}^T_{t=1}}[Y_Q(T) \mid T]$ can be written as
  \begin{align}
    & \mathbb E_{\{i(t)\}^T_{t=1}}[Y_Q(T) \mid T] \nonumber\\
    =& 1 + \sum^T_{t=d^M+1} \mathbb P(i(t)\in \mathcal N^-_Q \mid T) \nonumber\\
    \leq& 1 + \hspace{-2ex}\sum^T_{t=d^M+1} \hspace{-1ex}\left( \mathbb P(i(t) \in \lfloor \hat{\mathbf{x}}^*(t) \rfloor, i(t) \in \mathcal{N}^{-}_Q \mid T ) + \frac{\left(\frac{c_{max}}{c_{min}}\right)^2}{T-t+1} \right)  \nonumber\\
    \leq& 1 + \hspace{-2ex}\sum^T_{t=d^M+1} \hspace{-2ex}\mathbb P(i(t) \in \lfloor \hat{\mathbf{x}}^*(t) \rfloor, i(t) \in \mathcal{N}^{-}_Q \mid T ) + \left(\frac{c_{max}}{c_{min}}\right)^2 \ln T  \nonumber\\
    \leq& 1 + \hspace{-2ex}\sum^T_{t=d^M+1} \hspace{-2ex}\mathbb P(i(t) \in \lfloor \hat{\mathbf{x}}^*(t) \rfloor, i(t) \notin \lfloor \tilde{\mathbf{x}}^* \rfloor \mid T )   + \left(\frac{c_{max}}{c_{min}}\right)^2 \ln T  
  \end{align}

  We then derive the bound of the sum of the first two items by considering \textbf{Lemma}~\ref{le:swap} as follows.
  \begin{align}
    & 1 + \sum^T_{t=d^M+1} \mathbb P(i(t) \in \lfloor \hat{\mathbf{x}}^*(t) \rfloor, i(t) \notin \lfloor \tilde{\mathbf{x}}^* \rfloor \mid T ) \nonumber\\
    \leq& 1 + \sum^T_{t=d^M+1} \mathbb P \left( \frac{\bar r_Q(t-1) + \sqrt{\frac{2\log t}{\lambda_Q(t-1)}}}{c_{min}(\mathcal N^-_Q)} \geq \frac{\bar r_{Q'}(t-1) + \sqrt{\frac{2\log t}{\lambda_{Q'}(t-1)}}}{c_{max}(\mathcal N^+_{Q'})} ~\Bigg|~ T \right) \nonumber\\
    \leq& \ell + \sum^T_{t=d^M+1} \mathbb P \left( {\frac{\bar r_Q(t-1) + \sqrt{\frac{2\log t}{\lambda_Q(t-1)}}}{c_{min}(\mathcal N^-_Q)} \geq \frac{\bar r_{Q'}(t-1) + \sqrt{\frac{2\log t}{\lambda_{Q'}(t-1)}}}{c_{max}(\mathcal N^+_{Q'})},  \lambda_Q(t-1) \geq \ell} ~\Bigg|~ T \right) \nonumber\\
    \leq& \ell + \sum^T_{t=d^M+1} \mathbb P \left( {{ \max_{\ell \leq s_Q < t} \frac{\bar r_Q(t-1) + \sqrt{\frac{2\log t}{s_Q}}}{c_{min}(\mathcal N^-_Q)} \geq \min_{1 \leq s_{Q'} < t} \frac{\bar r_{Q'}(t-1) + \sqrt{\frac{2\log t}{s_{Q'}}}}{c_{max}(\mathcal N^+_{Q'})}} ~\Bigg|~ T} \right) \nonumber\\
    \leq& \ell + \sum^T_{t=1} \sum^{t-1}_{s_{Q'}=1} \sum^{t-1}_{s_{Q}=\ell} \mathbb P \left( \frac{\bar r_Q(t-1) + \sqrt{\frac{2\log t}{s_Q}}}{c_{min}(\mathcal N^-_Q)} \geq \frac{\bar r_{Q'}(t-1) + \sqrt{\frac{2\log t}{s_{Q'}}}}{c_{max}(\mathcal N^+_{Q'})} ~\Bigg|~ T \right) 
  \vspace{2ex}
  \end{align}
  If it holds that $\frac{\bar r_Q(t-1) + \sqrt{\frac{2\log t}{s_Q}}}{c_{min}(\mathcal N^-_Q)} \geq \frac{\bar r_{Q'}(t-1) + \sqrt{\frac{2\log t}{s_{Q'}}}}{c_{max}(\mathcal N^+_{Q'})}$, then at least one of the following three event must happen
  \begin{align}
    & \mathsf{Event1}: ~ \frac{\bar r_Q(t-1)}{c_{min}(\mathcal N^-_Q)} - \frac{\sqrt{\frac{2 \log t}{s_Q}}}{c_{min}(\mathcal N^-_Q)} \geq \frac{\mu_Q}{c_{min}(\mathcal N^-_Q)} \label{eq:event1}\\
    & \mathsf{Event2}: ~ \frac{\bar r_{Q'}(t-1)}{c_{max}(\mathcal N^+_{Q'})} - \frac{\sqrt{\frac{2 \log t}{s_{Q'}}}}{c_{max}(\mathcal N^+_{Q'})} \leq \frac{\mu_{Q'}}{c_{max}(\mathcal N^+_{Q'})} \label{eq:event2}\\
    & \mathsf{Event3}: ~ \frac{\mu_{Q'}}{c_{max}(\mathcal N^+_{Q'})} - \frac{\mu_{Q}}{c_{min}(\mathcal N^-_Q)} \leq \frac{2\sqrt{\frac{2\log t}{s_Q}}}{c_{min}(\mathcal N^-_Q)} \label{eq:event3}
  \end{align}

  By applying the Chernoff-Hoeffding inequalities \cite{DubhashiP-book09}, we have
  \begin{align*}
    \mathbb P(\mathsf{Event1}) = \mathbb P \left( \bar r_Q(t-1)- \sqrt{\frac{2 \log t}{s_Q}} \geq \mu_Q \right) \leq t^{-4}
  \end{align*}
  Similarly,
  \begin{align*}
    \mathbb P(\mathsf{Event2}) = \mathbb P \left( \bar r_{Q'}(t-1) - \sqrt{\frac{2 \log t}{s_{Q'}}} \leq \mu_{Q'} \right) \leq t^{-4}
  \end{align*}
  When $\ell = \left\lceil \frac{8\ln T}{c^2_{min} \delta^2_{min}} \right\rceil$, for any $s_Q = \ell, \ell+1, \cdots, t-1$, we have
  \begin{equation*}
    \frac{\mu_{Q'}}{c_{max}(\mathcal N^+_{Q'})} - \frac{\mu_{Q}}{c_{min}(\mathcal N^-_Q)} > \frac{2\sqrt{\frac{2\log t}{s_Q}}}{c_{min}(\mathcal N^-_Q)}
  \end{equation*}
  and thus $\mathbb P(\mathsf{Event3}) = 0$. Combining the above inequalities, we have
  \begin{align}
    & \mathbb E_{\{i(t)\}^T_{t=1}}[Y_Q(T) \mid T] \nonumber\\
    \leq& 1 + \hspace{-1ex} \sum^T_{t=d^M} \hspace{-1ex} \mathbb P(i(t) \in \lfloor \hat{\mathbf{x}}^*(t) \rfloor, i(t) \notin \lfloor \tilde{\mathbf{x}}^* \rfloor \mid T )   + \left(\frac{c_{max}}{c_{min}}\right)^2 \ln T  \nonumber\\
    \leq& \ell + \sum^T_{t=1} \sum^{t-1}_{s_{Q'}=1} \sum^{t-1}_{s_{Q}=\ell} \mathbb P \left(  \frac{\bar r_Q(t-1) + \sqrt{\frac{2\log t}{s_Q}}}{c_{min}(\mathcal N^-_Q)} \geq \frac{\bar r_{Q'}(t-1) + \sqrt{\frac{2\log t}{s_{Q'}}}}{c_{max}(\mathcal N^+_{Q'})} ~\Bigg|~ T \right) \nonumber\\
    \leq& \ell + \sum^T_{t=1} \sum^{t-1}_{s_{Q'}=1} \sum^{t-1}_{s_{Q}=\ell} \left( \mathbb P(\mathsf{Event1}) + \mathbb P(\mathsf{Event2}) + \mathbb P(\mathsf{Event3}) \right) \nonumber\\
    \leq& \left\lceil \frac{8\ln T}{c^2_{min} \delta^2_{min}} \right\rceil + \sum^T_{t=1} \sum^{t-1}_{s_{Q'}=1} \sum^{t-1}_{s_{Q}=\ell} 2t^{-4}  + \left( \frac{c_{max}}{c_{min}} \right)^2 \ln T + \frac{\pi^2}{3} + 1 \nonumber\\
    \leq& \xi \ln T + \frac{\pi^2}{3} + 1 \nonumber
  \end{align}
  where $\xi$ is defined in (25).
  \end{proof}

  Now, we are ready to prove the inequalities (\ref{eq:bdt}) and (\ref{eq:reg22}). Our CAWS algorithm proceeds until we have no more residual budget to select any workers such that
  \begin{equation}
    \mathbb P \left( \sum^T_{t=1} c_{i(t)} \geq B - c_{max} \right)=1
  \end{equation}
  therefore, we have
  \begin{align}
    & B-c_{max} \nonumber\\
    \leq& \mathbb E_{T, \{i(t)\}^T_{t=1}}\left[\sum^T_{t=1} c_{i(t)}\right] \nonumber\\
    =& \mathbb E_{T}\left[\sum^T_{t=1} \sum_{Q \in \Omega} \sum_{j \in \mathcal N_Q} c_j \mathbb P(i(t)=j \mid T) \right] \nonumber\\
    =& \mathbb E_{T}\left[\sum^T_{t=1} \left( c_{i^*} + \sum_{Q \in \Omega} \sum_{j \in \mathcal N_Q} (c_j - c_{i^*}) \mathbb P(i(t)=j \mid T) \right) \right] \nonumber\\
    \leq& \mathbb E_{T}\Bigg[ \sum^T_{t=1} \sum_{Q \in \Omega} \Bigg( \sum_{j \in \mathcal N^-_Q} (c_j - c_{i^*}) \mathbb P(i(t)=j \mid T) + \sum_{j \in \mathcal N^+_Q} (c_j - c_{i^*}) \mathbb P(i(t)=j \mid T) \Bigg) \Bigg] + \mathbb E_{T}[T] c_{i^*}  \nonumber\\
    \leq& \mathbb E_{T}\Bigg[ \sum^T_{t=1} \sum_{Q \in \Omega} \Bigg( \sum_{j \in \mathcal N^-_Q} (c_{max}(\mathcal N^-_Q) - c_{i^*}) \mathbb P(i(t)=j \mid T) + \sum_{j \in \lfloor \tilde{\mathbf x}^* \rfloor} (c_j - c_{i^*}) \mathbb P(i(t)=j \mid T)\Bigg) \Bigg] + \mathbb E_{T}[T] c_{i^*}  \nonumber\\
    \leq& \mathbb E_{T}[T] c_{i^*} + \mathbb E_T \left[ \sum_{j \in \lfloor \tilde{\mathbf x}^* \rfloor} (c_j - c_{i^*}) \mathbb E_{\{i(t)\}^T_{i=1}}[x_j \mid T]  \right]  \nonumber\\
    &+ \mathbb E_{T}\left[ \sum_{Q: c_{max}(\mathcal N^-_Q) > c_{i^*}} \left( c_{max}(\mathcal N^-_Q \right) - c_{i^*}) \mathbb E_{\{i(t)\}^T_{i=1}}[Y_Q(T) \mid T] \right]
  \end{align}
  and thus
  \begin{align}
    &\mathbb E_T[T] \nonumber\\
    \geq& \frac{B-c_{max}}{c_{i^*}} - \mathbb E_T \left[ \sum_{j \in \lfloor \tilde{\mathbf x}^* \rfloor} \frac{c_j - c_{i^*}}{c_{i^*}} \mathbb E[x_j \mid T]  \right]  \nonumber\\
    & - \mathbb E_{T} \left[ \sum_{Q: c_{max}(\mathcal N^-_Q) > c_{i^*}} \frac{c_{max}(\mathcal N^-_Q) - c_{i^*}}{c_{i^*}} \mathbb E_{\{i(t)\}^T_{i=1}}[Y_Q(T) \mid T] \right]
  \end{align}
  The validity of (\ref{eq:bdt}) can be proved by substituting (\ref{eq:upbdy1}) (see \textbf{Lemma}~\ref{le:upbdy}) into the second item on the right side of the above inequality such that
  \begin{align}
    & \mathbb E_{T}\left[ \sum_{Q: c_{max}(\mathcal N^-_Q) > c_{i^*}} \hspace{-3ex} \frac{c_{max}(\mathcal N^-_Q) - c_{i^*}}{c_{i^*}} \mathbb E_{\{i(t)\}^T_{i=1}}[Y_Q(T) \mid T] \right] \nonumber\\
    \leq& \sum_{Q : c_{max}(\mathcal N^-_Q) > c_{i^*}} \hspace{-3ex}  \frac{c_{max}(\mathcal N^-_Q) - c_{i^*}}{c_{i^*}} \left( \xi \cdot \mathbb E_T [\ln T] + \frac{\pi^2}{3} + 1 \right) \nonumber\\
    \leq& \sum_{Q : c_{max}(\mathcal N^-_Q) > c_{i^*}} \hspace{-3ex}  \frac{c_{max}(\mathcal N^-_Q) - c_{i^*}}{c_{i^*}} \cdot h(\ln B)
  \end{align}
  where the second inequality holds since $T \leq \frac{B}{c_{min}}$.

  We finally prove that the inequality (\ref{eq:reg22}) hold by the derivation as follows
  \begin{eqnarray} 
    && \mathbb E_T \left[ T\mu_{Q_{i^*}} - \sum_{Q \in \Omega} \sum_{i \in \mathcal N_{Q}} \mu_Q \mathbb E_{\{i(t)\}^T_{t=1}}[x_i\mid T] \right] \nonumber\\
    &=& \mathbb E_T \left[  \sum^T_{t=1} \sum_{Q \in \Omega} (\mu_{Q_{i^*}} - \mu_Q) \mathbb P( i(t) \in \mathcal N_Q \mid T)  \right] \nonumber\\
    &=& \mathbb E_T \bigg[ \sum^T_{t=1} \sum_{Q \in \Omega} (\mu_{Q_{i^*}} - \mu_Q) \mathbb P( i(t) \in \mathcal N^-_Q \mid T) - \sum^T_{t=1} \sum_{Q \in \Omega} (\mu_{Q_{i^*}} - \mu_Q) \mathbb P( i(t) \in \mathcal N^+_Q \mid T)  \bigg] \nonumber\\
    &=& \mathbb E_T \bigg[ \sum^T_{t=1} \sum_{Q \in \Omega} (\mu_{Q_{i^*}} - \mu_Q) \mathbb P( i(t) \in \mathcal N^-_Q  \mid T)  - \sum^T_{t=1} \sum_{j \in \lfloor \tilde{\mathbf x}^* \rfloor} (\mu_{Q_{i^*}}-\mu_{Q_j}) \mathbb P( i(t) =j \mid T )  \bigg] \nonumber\\
    &\leq& \mathbb E_T \bigg[ \sum^T_{t=1} \sum_{Q: \mu_{Q_{i^*}} > \mu_Q} (\mu_{Q_{i^*}} - \mu_Q) \mathbb P( i(t) \in \mathcal N^-_Q \mid T)  - \sum^T_{t=1} \sum_{j \in \lfloor \tilde{\mathbf x}^* \rfloor} (\mu_{Q_{i^*}}-\mu_{Q_j}) \mathbb P( i(t) =j \mid T )  \bigg] \nonumber\\
    &=& \mathbb E_T \left[ \sum_{Q: \mu_{Q_{i^*}} > \mu_Q} (\mu_{Q_{i^*}} - \mu_Q) \mathbb E_{\{i(t)\}^T_{t=1}}[Y_Q(T) \mid T] \right] -  \mathbb E_T \left[ \sum_{j \in \lfloor \tilde{\mathbf x}^* \rfloor} (\mu_{Q_{i^*}}-\mu_{Q_j}) \mathbb E_{\{i(t)\}^T_{t=1}}[x_j \mid T] \right]  \nonumber\\
    &\leq& \sum_{Q: \mu_{Q_{i^*}} > \mu_Q} (\mu_{Q_{i^*}} - \mu_Q)  h(\ln B)  -  \mathbb E_T \left[ \sum_{j \in \lfloor \tilde{\mathbf x}^* \rfloor} (\mu_{Q_{i^*}}-\mu_{Q_j}) \mathbb E_{\{i(t)\}^T_{t=1}}[x_j \mid T] \right]  
  \end{eqnarray}
  where we have last inequality by considering \textbf{Lemma}~\ref{le:upbdy} again.

\end{appendices}

\end{document}